\documentclass[prodmode]{acmsmall}

\usepackage{times}
\usepackage{latexsym,epsf,epsfig}
\usepackage{amssymb,amsmath}

\DeclareMathOperator*{\argmin}{argmin}
\DeclareMathOperator*{\argmax}{argmax}
\newcommand*{\argminl}{\argmin\limits}
\newcommand*{\argmaxl}{\argmax\limits}

\makeatletter
\newif\if@restonecol
\makeatother

\usepackage{balance}
\usepackage[ruled,linesnumbered,vlined]{algorithm2e}
\usepackage{enumitem}

\begin{document}

\markboth{A. Nazi et al.}{Web Item Reviewing Made Easy By Leveraging Available User Feedback}

\title{Web Item Reviewing Made Easy By Leveraging Available User Feedback}
%


\author{AZADE NAZI
\affil{Department of Computer Science and Engineering, University of Texas at Arlington, USA}
MAHASHWETA DAS
\affil{Hewlett Packard Labs, Palo Alto, USA}
GAUTAM DAS
\affil{Department of Computer Science and Engineering, University of Texas at Arlington, USA}}

\begin{abstract} 
The increasing popularity and widespread use of online review sites over the past decade has motivated businesses of all types to possess an expansive arsenal of user feedback (preferably positive) in order to mark their reputation and presence in the Web. Though a significant proportion of purchasing decisions today are driven by average numeric scores (e.g., movie rating in IMDB), detailed reviews are critical for activities such as buying an expensive digital SLR camera, reserving a vacation package, etc. Since writing a detailed review for a product (or, a service) is usually time-consuming and may not offer any incentive, the number of useful reviews available in the Web is far from many. The corpus of reviews available at our disposal for making informed decisions also suffers from spam and misleading content, typographical and grammatical errors, etc. In this paper, we address the problem of how to engage the lurkers (i.e., people who read reviews but never take time and effort to write one) to  participate and write online reviews by systematically simplifying the reviewing task. Given a user and an item that she wants to review, the task is to identify the top-$k$ meaningful phrases (i.e., tags) from the set of all tags (i.e., available user feedback for items) that, when advised, would help her review an item easily. We refer to it as the TagAdvisor problem, and formulate it as a general-constrained optimization goal. Our framework is centered around three measures - relevance (i.e., how well the result set of tags describes an item to a user), coverage (i.e., how well the result set of tags covers the different aspects of an item), and polarity (i.e., how well sentiment is attached to the result set of tags) in order to help a user review an item satisfactorily. By adopting different definitions of coverage, we identify two concrete problem instances that enable a wide range of real-world scenarios. We show that these problems are NP-hard and develop practical algorithms with theoretical bounds to solve them efficiently. We conduct detailed experiments on synthetic and real data crawled from the web to validate the utility of our problem and effectiveness of our solutions.
\end{abstract}




\begin{bottomstuff}
The work of Azade Nazi and Gautam Das was partially supported by National Science Foundation under grants 0915834, 1018865, Army Research Office under grant W911NF-15-1-0020 and a grant from Microsoft Research. Any findings, conclusions, or recommendations expressed in this material are those of the authors and do not necessarily reflect the views of the sponsors listed above.

The work of Mahashweta Das is done while she was a graduate student at the University of Texas at Arlington.

Contact A. Nazi at {\em azade.nazi@mavs.uta.edu}

Contact M. Das at {\em mahashweta.das@hpe.com}

Contact G. Das at {\em gdas@cse.uta.edu}

\end{bottomstuff}

\maketitle

\section{Introduction}\label{sec:intro}

The increasing popularity and widespread use of online reviews in sites like Yelp, Amazon, Angie's List, TripAdvisor, etc. over the past decade has motivated businesses of all types to possess an expansive arsenal of user feedback (preferably positive) in order to mark their reputation and presence in the Web. User feedback is available in various forms such as numeric or star ratings, number of visits, number of check-ins, number of Facebook likes, tags, reviews, etc. Though a significant proportion of purchasing decisions today are driven by aggregate user feedback in the form of average rating (e.g., a movie in IMDB), number of Facebook check-in (e.g., a restaurant page in Facebook), number of views (e.g., an article in Business Insider), etc., detailed reviews continue to influence a wide variety of critical activities such as buying an expensive digital SLR camera, choosing a car, reserving a vacation package, etc. However, since writing a detailed review for a product (or, a service) is usually time-consuming and may not offer any incentive, the number of useful reviews available is far from many. Though the 1\% rule (or, the 90-9-1 rule) of Internet is presumed to be dead, the proportion of lurkers (i.e., people who read user-generated content in the Web without contributing) is still high. According to survey conducted by BrightLocal in 2015, though 92\% of consumers read online reviews, word of mouth is still the most popular way to recommend a product or a service.
Moreover, several sites like Yelp and IMDB allows users to submit feedback as ratings without any review accompaniment. As a result, the number of numerical ratings available for a product far exceeds the number of detailed reviews. The corpus of reviews available at our disposal for making informed decisions suffers from redundancy, inaccurate and misleading content, typographical and grammatical errors, etc. too.     

In this paper, we investigate how to engage the users to participate and write online reviews by systematically simplifying the web item (e.g., electronic products, apparel, restaurants, movies, music, travel itineraries, etc.) reviewing task. Given user feedback for items by past users in the form of text, a user and an item that she wants to review, our objective is to identify a set of {\em meaningful} phrases (i.e., {\em tags}) that we {\em advise} to the user in order to help her review the item. We refer to this as the {\bf TagAdvisor} problem. The user would quickly choose from among the set of returned tags to articulate her feedback for the item without having to spend a lot of time {\em writing} the review.

The top-$k$ tags should not only meet the necessary requisites of a {\em good} online review like conciseness, comprehensiveness, objectiveness, etc. but should also offer adequate incentive in the form of simple usability, easy applicability, etc. 
As one of our first step towards solution, we employ state-of-art text mining techniques (discussed later in Section~\ref{sec:expt}) and extract meaningful phrases or tags from user feedback in the form of text, i.e., reviews.
Since each tag is a user feedback for an item, the tags are extracted with sentiment labels attached to it. $T^+$ and $T^-$ are the set of positive and negative tags respectively. For example, a review statement ``{\em It is a lightweight camera with some amazing features}'' is reduced to the tags \{{\tt lightweight camera}, {\tt amazing features}\}, where both tags have positive sentiment. 

We formulate the problem of identifying the top-$k$ meaningful tags from the set of all tags (i.e., available user feedback for items) for a user-item pair, as a novel general-constrained optimization problem. A core challenge in this design is defining the essential properties of the top-$k$ tags to be returned that would serve to review the item effectively. We consider {\em relevance} (i.e., how well the result set of tags describes an item to a user), {\em coverage} (i.e., how well the result set of tags covers the diverse aspects of an item), and {\em polarity} (i.e., how well sentiment is attached to the result set of tags) in order to enable a user to satisfactorily review an item. Though relevance and coverage have been studied in the past~\cite{Belem:2013}, our work is the first to consider all three measures simultaneously in the context of tag mining.

A user can review an item in different ways. A user can express her broad opinion about the different aspects of an item which, in turn, can either be positive or negative. Again, a user can express both positive and negative opinion for the same attribute (or, set of attributes) of the item. For example, a user may write a review for a camera as  {\em ``The picture quality of this camera is great and so is the sharpness and color accuracy of the pictures, but the battery life is short.''}, while another user of the same camera may write {\em ``Though the extra screen with touchscreen and gesture-control features saps battery life, it's perfect for fashion-conscious snap shooters.''}. The first review contains positive feedback for the camera's image quality and negative feedback for the camera's battery life. The second review contains both positive and negative feedback for the camera's advanced features \{dual-screen, touchscreen and gesture-control\}. Therefore, the item attributes that were covered by the review is independent of the feedback sentiment in the former case, and dependent on the sentiment in the latter. This motivates us to propose two problem instances, namely \textbf{Independent Coverage TagAdvisor} problem and \textbf{Dependent Coverage TagAdvisor} problem that considers two different definitions of coverage respectively in order to satisfy users' real world needs. 
 


Though the output of our problem is recommending a set of tags for a user-item pair, our objective is different from the literature of work dedicated to {\em tag recommendation}~\cite{DBLP:conf/kdd/FengW12,DBLP:conf/sigir/SongZLZLLG08}. The top-$k$ tags in our problem are more feedback than descriptive relevant information for an item and hence calls for additional properties like coverage of all aspects of the item in order to ensure diversity, as well as sentiment polarity in opinion of the user for the different aspects of the item.
The latter deals with the automated process of suggesting useful and informative tags to an emerging resource based on historical information in order to help search, exploration, and navigation. In our problem, the tags are more feedback than information about the resource and hence calls for additional properties like coverage of all aspects of the item in order to ensure diversity as well as sentiment polarity in opinion of the user for the different aspects of the item. 
While {\em review summarization}, that helps users read the valuable content in the vast volumes of user feedback for items, has been researched in the literature~\cite{DBLP:conf/ACMicec/GhoseI07,DBLP:conf/kdd/HuL04,DBLP:conf/kdd/LappasCT12,DBLP:conf/kdd/TsaparasNT11}, our objective of simplifying a user's review writing task has not been studied to the best of our knowledge. Moreover, none of the existing work on review summarization, ranking, and selection accommodate relevance, coverage, and polarity that we consider in our framework. Even {\em collaborative filtering} based approaches for tag recommendation consider only relevance measure to determine the top-$k$ tags~\cite{DBLP:conf/pkdd/JaschkeMHSS07}. 

The TagAdvisor (TA) problem is technically challenging for several reasons. Our objective is to identify k tags that are relevant, cover different aspects of an item, and have well-balanced positive and negative sentiment attached to it. While the first two concerns the relationship between the item attributes and tags, the third is dependent on a user's personal preference. Some users tend to be lenient and provide mostly positive feedback; some tend to be critical. In this paper, we choose to focus on modeling the complex dependencies that exist between item attributes and tags and leverage user personal preference as a parameter, thereby letting the system deal with both new users and with new items, alleviating cold-start problems. Classifiers and rule learning techniques in the literature can be used to predict the relevance of tags for an item. In this paper, we employ existing techniques to predict the rules modeling the relationship between attributes and tags, where each rule has a probability of occurrence. 

As discussed earlier, formalizing the users' different ways of reviewing an item relates to the coverage characteristic of the top-$k$ tags to be returned. By adopting different definitions of coverage, we propose problems that enable a wide range of real-world scenarios. For a user reviewing an item, the Independent Coverage TagAdvisor (IC-TA) problem identifies top-k tags that are relevant, satisfy the user's criticalness in reviewing, and maximizes the number of item attributes covered by them, independent of their sentiment. On the other hand, the Dependent Coverage TagAdvisor (DC-TA) problem returns tags that cover item attributes both positively and negatively, in addition to being relevant and satisfying user's criticalness in reviewing. As one of our first results, we show that each of these problem is NP-Complete by reduction from  Max-Coverage problem with Group Budget Constraints problem and MAX-SUM Facility Dispersion problem respectively. Given this intractability result, designing efficient algorithmic solutions that work well in practice is challenging. In addition, the objective function of the second problem is proved to be not sub-modular thereby precluding the direct use of off-the-shelf greedy algorithms. For each problem, we develop two algorithmic solutions yielding optimal solutions, namely: (a) brute-force naive methods (\textbf{E-IC-TA} and \textbf{E-DC-TA}) and (b) techniques based on Integer Linear Programming (ILP) methods (\textbf{ILP-IC-TA} and \textbf{ILP-DC-TA}), which work well for moderate-sized problem instances. We also developed efficient algorithms that yield approximate solutions (\textbf{A-IC-TA} and \textbf{A-DC-TA}). We prove that each of our approximation algorithm produces solution with constant approximation factor. We conduct experiments on synthetic data and real data crawled from Yahoo! Autos, Walmart and Google Product to evaluate the efficiency and quality of our proposed algorithms. We present an Amazon Mechanical Turk user study and an interesting case study on real camera data to validate the effectiveness of our solution over that by state-of-art.

\vspace{0.02in}
\noindent
In summary, we make the following main contributions:
\vspace{-0.01in}
\begin{itemize}
\item We introduce and motivate the novel \textbf{TagAdvisor} problem that leverages available user feedback for items in online review sites to simplify the review writing task. Our objective is to identify the top-$k$ meaningful tags that, when advised to a user, would help her review an item easily.
\item We formulate the problem as a general-constrained optimization goal. Our formulation is centered around three measures \textemdash relevance, coverage, and polarity.
\item We formalize the users' different ways of reviewing an item by proposing two coverage functions and thereby defining two concrete problem instances, namely Independent Coverage TagAdvisor (\textbf{IC-TA}) and Dependent Coverage TagAdvisor (\textbf{DC-TA}) problems, that enable a wide range of real-world scenarios.
\item We show that each of the problems is NP-Complete and develop optimal Integer Linear Programming (ILP) based algorithms and practical algorithms with compelling theoretical properties to solve them efficiently.
\item We perform detailed experiments on synthetic and real data crawled from the web to demonstrate the utility of our problem and effectiveness of our algorithms.
\end{itemize}

\begin{table}
\centering
\tbl{An example camera review data as triple $<U,I,T>$\label{tbl:samsungFeatures}}{
\centering
\begin{tabular}{|p{0.6cm}|p{0.35cm}|p{0.85cm}|p{1.0cm}||p{1.0cm}|p{1.1cm}|p{0.7cm}|p{0.5cm}|p{0.5cm}|p{0.45cm}|p{0.7cm}|p{0.75cm}|p{1.0cm}||p{2.4cm}|} 
\hline
\multicolumn{4}{ |c|| }{\textbf{Users (U)}} & \multicolumn{9}{ c|| }{\textbf{Items (I)}} & \multicolumn{1}{ c| }{\textbf{Tags (T)}}\\ 
\hline
 {\sf \footnotesize User Name} & {\sf \footnotesize Age} & {\sf \footnotesize Gender} & {\sf \footnotesize  Location} & {\sf \footnotesize Item Name } & {\sf \footnotesize Resolution} & {\sf \footnotesize Optical Zoom } & {\sf \footnotesize Color} & {\sf \footnotesize Front LCD} & {\sf \footnotesize Back LCD} & {\sf \footnotesize Shutter Speed} & {\sf \footnotesize Touch screen} & {\sf \footnotesize Gesture Control} & {\sf \footnotesize Tags}\\
$\sf (u)$ & $\sf (c_1)$ & $\sf (c_2)$ & $\sf (c_3)$ & $\sf (i)$ & $\sf (a_1)$ & $\sf (a_2)$ & $\sf (a_3) $ & $\sf (a_4)$ & $\sf (a_5)$ & $\sf (a_6)$ & $\sf (a_7)$ & $\sf (a_8)$ & $\sf (T)$\\ 
\hline
\hline 
 {\sf \footnotesize Amy} & {\sf \footnotesize 23} & {\sf \footnotesize Female} & {\sf \footnotesize California} & {\sf \footnotesize Samsung TL225} & {\sf \footnotesize 12.2mp} & {\sf \footnotesize 4.6x} & {\sf \footnotesize Red} & {\sf \footnotesize 1.5"} & {\sf \footnotesize 3.5"} & {\sf \footnotesize 8-1/2000} & {\sf \footnotesize true} & {\sf \footnotesize true} & {\tt \small super cool}, {\tt \small stylish}, {\tt \small poor battery life}, {\tt \small lightweight}\\
\hline
 {\sf \footnotesize David} & {\sf \footnotesize 35} & {\sf \footnotesize Male} & {\sf \footnotesize Ohio} &{\sf \footnotesize Samsung TL225} & {\sf \footnotesize 12.2mp} & {\sf \footnotesize 4.6x} & {\sf \footnotesize Red} & {\sf \footnotesize 1.5"} & {\sf  \footnotesize 3.5"} & {\sf \footnotesize 8-1/2000} & {\sf \footnotesize true} & {\sf \footnotesize true} & {\tt \small poor battery life}, {\tt \small blurry pictures}, {\tt \small gimmicky touchscreen}\\
\hline 
\end{tabular}}
\end{table}

\begin{table}
\tbl{Set of rules for example data in Table~\ref{tbl:samsungFeatures}\label{tbl:samsungRules}}{
\centering
\begin{tabular}{|p{2.8cm}|p{7.4cm}|p{0.2cm}|p{3.25cm}|p{1.25cm}|p{0.4cm}|} 
\hline
$\{{\sf a}\}$ & {\sf \small Attributes} & ${t_x}$ & ${\sf \small Tags}$ & {\sf \small Sentiment} & $p$\tabularnewline
\hline
\hline 
{\sf \small \{$a.v_4, a.v_7, a.v_8$\}} & {\sf \footnotesize Front LCD=1.5"}, {\sf \small Touchscreen=true}, {\sf \footnotesize Gesture Control=true} & $t_1$ & \tt \small {super cool} & + & 0.3\tabularnewline
\hline
{\sf \small \{$a.v_3, a.v_4 , a.v_7, a.v_8$\}} & {\sf \footnotesize Color=Red}, {\sf \footnotesize Front LCD=1.5"}, {\sf \footnotesize Touchscreen=true}, {\sf \footnotesize Gesture Control=true}  & $t_2$ & \tt \small {stylish} & + & $0.2$ \tabularnewline
\hline
{\sf \small \{$a.v_1, a.v_2, a.v_5$\}} & {\sf \footnotesize Resolution=12.2mp}, {\sf \footnotesize Optical Zoom=4.6x}, {\sf \footnotesize Back LCD=3.5"} & $t_3$ & \tt \footnotesize {lightweight} & + & $0.1$\\
\hline
{\sf \small \{$a.v_4, a.v_7, a.v_8$\}} & {\sf \footnotesize Front LCD=1.5"}, {\sf \footnotesize Touchscreen=true},{\sf \footnotesize Gesture Control=true} & $t_4$ & \tt \small {poor battery life} & - & $0.13$\tabularnewline
\hline
{\sf \small \{$a.v_1, a.v_2, a.v_6$\}} & {\sf \footnotesize Resolution=12.2mp}, {\sf \footnotesize Optical Zoom=4.6x}, {\sf \footnotesize Shutter Speed=8-1/2000} & $t_5$& \tt \small {blurry pictures} & - & $0.12$\tabularnewline
\hline
{\sf \small \{$a.v_5, a.v_7, a.v_8$\}} & {\sf \footnotesize Back LCD=3.5"}, {\sf \footnotesize Touchscreen=true}, {\sf \footnotesize Gesture Control=true} & $t_6$ & \tt \small {gimmicky touchscreen}  & - & $0.15$\tabularnewline
\hline 
\end{tabular}}
\end{table}

\section{The T\lowercase{ag}A\lowercase{dvisor} Framework}
\label{sec:model}

\subsection{Preliminaries}
\label{subsec:preliminaries}
We model the data $D$ in an online review site as a triple $<U, I, T>$, representing the sets of users, items, and the tag vocabulary respectively. Let $n$ be the total number of tags in $T$. Each tagging action can be considered as a triple itself, represented as  $<u, i, {\tt T}>$ where $u \in U$, $i \in I$, and ${\tt T} \in T$. We assume that each user $u \in U$ has a well-defined schema $U_A = \{c_1, c_2, ...\}$, where the attributes typically are the demographic information such as  {\small \textsf{name}, \textsf{age}, \textsf{gender}, \textsf{location}}, etc. A user $u$ is represented as a tuple $\{c.v_1, c.v_2, ...\}$ conforming to $U_A$, where $c.v_y$ is the value of the user attribute $c_y$; e.g., $<${\small \textsf{name=Amy}, \textsf{age=23}, \textsf{gender=Female}, \textsf{location=California}}$>$ represents a 23 years old female from California. Similarly, every item $i\in I$ is associated with a well-defined schema $I_A = \{a_1, a_2, ..., a_m\}$ and each item $i$ is a tuple $\{a.v_1, a.v_2, ..., a.v_m\}$ with $I_A$ as schema, where $a.v_y$ is the value of item attribute $a_y$; e.g., $<${\small \textsf{brand=Samsung}, \textsf{model=TL225}, \textsf{type=point and shoot}}$>$ describes a compact Samsung camera. Note that, our work is not influenced by or biased towards any brand. 
Since each tag is a user feedback for an item, it describes the item positively or negatively. Therefore, we partition $T$ into $T^+$ and $T^-$, where $\left\vert T^+\right\vert$ is $n^+$ and $\left\vert T^-\right\vert$ is $n^-$. 

\vspace{0.05in}
\noindent \textsc{Example}: \textit{Suppose, we would like to help a user review a camera, say Samsung TL225. Table~\ref{tbl:samsungFeatures} describes the data available in an online review site where users Amy and David have left tag-based feedback for the camera. Table~\ref{tbl:samsungFeatures} also shows the attribute values for the users and the camera. The set of all tags $T$ = \{ ${\tt T}_1$, ${\tt T}_2$\} for  item $i$ (i.e., Samsung TL225) by users $u_1$ (i.e., Amy) and $u_2$ (i.e., David) is classified into} ${T}^+$ = \{{\tt super~cool, stylish, lightweight}\} \textit{and} ${T}^-$ = \{{\tt blurry pictures, gimmicky touchscreen, poor battery life}\} \textit{by domain experts.}

Given an item $i$ and set of tags $T$, probabilistic classifiers can be used to compute the relevance of the tags for the item (i.e., $Pr (t_x | i)$). In this paper, we use the rule based classifiers~\cite{Cohen:1995,Liu:1998} to find the dependency of the item attributes to the tags and generate rules with probability of occurrence $p$, i.e., the relevance score. However, there exist a number of prior work that show popular classifiers like decision tree, random forest and SVM can also be used to generate rules~\cite{Quinlan:1987,Chiang:2001,Sirikulviriya:2011,Nunez:2002,Costa:2005,Barakat:2010,Diederich:20088}. We discuss the detail of the related work in Section~\ref{sec:relwork}.


\vspace{0.05in}
\noindent \textsc{Example} [continued]: \textit{Table~\ref{tbl:samsungRules} presents a set of rules associated with the Samsung TL225 camera and tags in Table~\ref{tbl:samsungFeatures}. Illustrating one of the rules:} \{{\small \sf Front LCD=1.5'', Touchscreen=true, Gesture Control=true}\} $\rightarrow$ {\tt short battery life} \textit{with $p = 0.13$ indicates that with probability of 0.13 the camera's dual LCD feature along with its touchscreen and gesture control interfaces are responsible for the camera receiving the tag {\tt short battery life}.}

For an item $i$ having attributes values \{$a.v_1, a.v_2, ... a.v_m$\}, if there are several rules for a tag $t_x$, the one with highest probability $p$ would be selected. For the rest of the paper, we use the example in Tables~\ref{tbl:samsungFeatures} and \ref{tbl:samsungRules} as the running example.

In this paper, our objective is to identify the top-$k$ tags $T^* = \{t_1, t_2, ...,t_k\}$ for a user $u \in U$ and an item $i \in I$ such that $u$ can review $i$ by choosing from $T^*$. The result set $T^*$ is selected from the tag vocabulary $T$ if they are ``meaningful''. Before formalizing the problem, let us define the essential characteristics that tags in $T^*$ must satisfy:   

\vspace{0.05in}
\noindent\textbf{Relevance}: \textit{Given item $i$ and tag vocabulary $T$, the relevance of a tag $t_x \in T^*$ denotes how well $t_x$ describes $i$. Mathematically, it is measured as the probability of obtaining $t_x$ given $i$, i.e., {\sc rel}($t_x, i$) = $Pr (t_x | i)$. As we have discussed earlier this score can be computed by employing a classifier modeling the relationship between item attributes and tags. Thus, {\sc rel}$(T^*) = FUNC_{t_x \in T*}\Big(\textsc {rel}(t_x,i)\Big)$ = $\sum_{t_x \in T^*}\Big(\textsc {rel}(t_x,i)\Big)$.}

Given a list of tags $T$ which is sorted by the relevance (i.e., {\sc rel}($t_x, i$) = $Pr (t_x | i)$), the maximum relevance score is the total score for the top $k$ tags in the sorted list. We represent the maximum relevance score for a set of $k$ tags from $n$ tags in $T$ as $\textsc{rel}_{max}^{T, k}$.

\vspace{0.05in}
\noindent\textbf{Coverage}: \textit{Given item $i$, tag vocabulary $T$, and a set of associated rules $\Re = \{\{a.v\}\rightarrow t_x \}$, the coverage of a tag $t_x \in T^*$ for $i$ is the set of distinct item attribute values have been covered by it. We say $t_x$ covers the attribute value $a.v_y$ if $a.v_y \in \{a.v\}$, i.e., {\sc cov}($t_x, i$) = $\{a.v\}$. Therefore, {\sc cov}$(T^*) = FUNC_{t_x \in T*}\Big(\textsc {cov}(t_x,i)\Big)$. We will discuss FUNC in details later in Section~\ref{subsec:concrete}.
} 

\vspace{0.05in}
\noindent\textbf{Polarity}: \textit{Given item $i$, and tag vocabulary $T$, the polarity of $T^*$ for a user reviewing item $i$ captures the distribution of sentiment in opinion. It is measured as the ratio of the number of the positive tags to the number of the negative tags, i.e., {\sc pol($T^*$)} = $\frac{|T^*{^{+}}|}{|T^*{^{-}}|}$.}

While maximization of the first two characteristics, i.e., relevance and coverage, for determining the set $T^*$ of top-$k$ tags is obvious, the third characteristics, i.e., polarity is dependent on a user's personal preference. Some users tend to be lenient and provide mostly positive feedback; some tend to be harsh. Thus, there is not any obvious way of estimating a user's criticalness in reviewing. One reasonable solution is to aggregate sentiments of user demographic groups and consider the value of the group to which the user belongs as her reviewing tendency. For example, if the average rating for cameras by all young female users living in California is 8.0 (on a scale of 10.0), then a user belonging to the sub-population will have a criticalness factor of 0.8 (on a 0-1 scale); she is likely to assign 80\% positive feedback and 20\% negative feedback to a camera. {\sc pol($T^*$)} = $\frac{|T^*{^{+}}|}{|T^*{^{-}}|}$ should be at least $\frac{0.8}{0.2}$, i.e., 4. In other words, polarity is the ``odds" of the positive tags which is the probability of positive tags $\frac{|T^*{^{+}}|}{|T^*|}$ to the probability of negative tags $\frac{|T^*{^{-}}|}{|T^*|}$. Since our TagAdvisor problem focuses on modeling the relationship between item attributes and tags, we leverage user personal preference as a parameter in our framework. We refer to this parameter, denoted by $\alpha$ as \textbf{User Factor}, where the value of the $\alpha$ is normalized to a [0,1] continuous sentiment scale.  
  

\subsection{The Problem}
\label{sec:problemDefinition}
A user can review an item in different ways. A user can express her opinion on multiple item attributes which in turn, can either be positive or negative. For example, the set of tags \{{\tt great picture quality, great sharpness, great color accuracy, short battery life}\} 
contains positive feedback for the camera's image quality and negative feedback for the camera's battery life. Again, a user can express both positive and negative opinion for the same attribute (or, set of attributes). For example, the set of tags \{{\tt short battery life}, {\tt stylish}\} 
contains both positive and negative feedback for the camera's innovative/advanced aspects (i.e.,dual-screen, touchscreen and gesture-controlled). From Table~\ref{tbl:samsungRules}, {\tt short battery life} and {\tt stylish} are tags related to camera attributes {\small \sf Front LCD}, {\small \sf Touchscreen} and {\small \sf Gesture Control} for Samsung TL225. 

We first propose a general TagAdvisor problem and then present two different problem instances that enable a wide range of
real-world scenarios. The instances are distinct by the difference in formulation of the coverage of a set of tags $T^*$, i.e., {\sc cov}$(T^*)$. 

\begin{figure}
\centerline{
\includegraphics[height=50mm,width = 80mm]{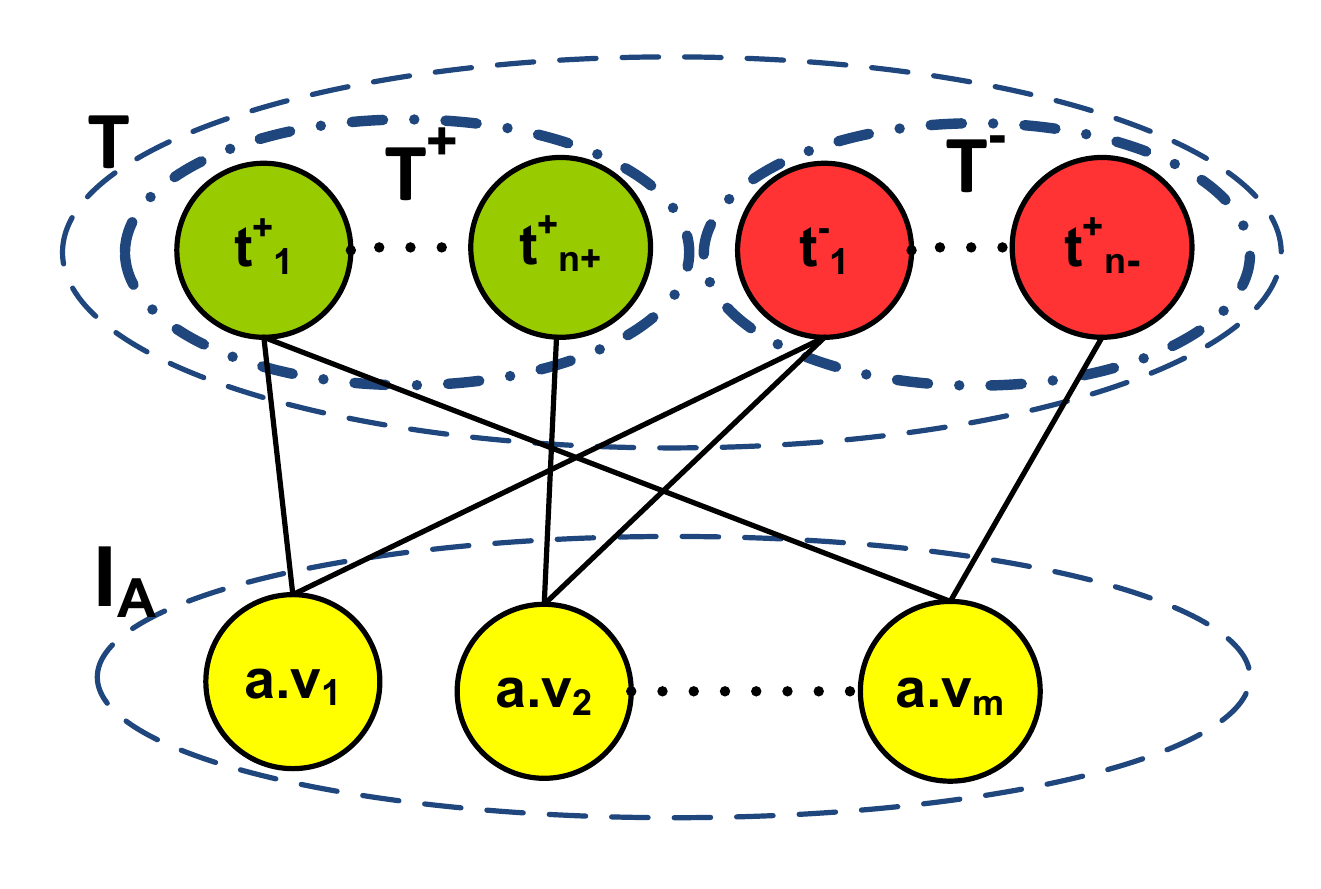}}
\caption{{TagAdvisor Bipartite Graph model.}}
\label{fig:TAModel}
\end{figure}


\begin{figure}
\centerline{
\includegraphics[height=50mm, width = 120mm]{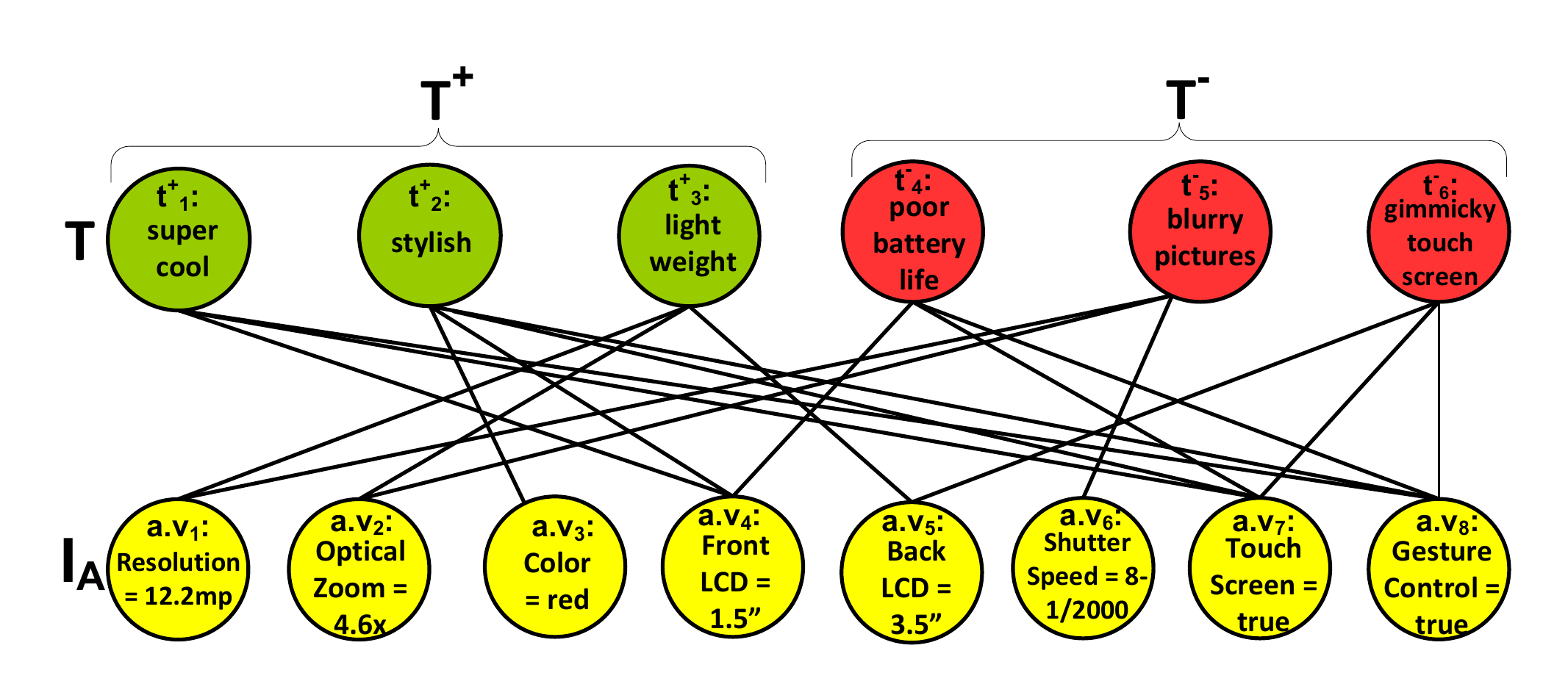}}
\caption{{TagAdvisor Bipartite Graph model of the Running Example.}}
\label{fig:TAModel-example}
\end{figure}

\vspace{0.05in}
\noindent\textsc{Definition 1.} 
\textbf{TagAdvisor Problem (TA)}: \textit{Given a set of rules $\Re = \{\{a.v\}\rightarrow t_x \}$ for an item $i= \{a.v_1, a.v_2, ...\}$ and $t_x \in T$, non-negative integer budget $k$, relevance parameter $\beta$ ($0 \leq \beta \leq 1$), and user factor $\alpha$ ($0 \leq \alpha \leq 1$), find a subset of $T^* \subseteq T$ such that:
\begin{list}{$\bullet$}{}
\item {$|T^*| \leq k$;}
\item {\sc pol$(T^*) = \frac{\alpha}{1-\alpha}$};
\item $\textsc{rel}(T^*) \geq \beta \times \textsc{rel}_{max}^{T,k}$;
\item {\sc cov}$(T^*)$ is maximized, 
\end{list}}

\noindent where {\sc pol($T^*$)} is the sentiment in opinion by tags in $T^*$, i.e., the number of positive tags ($k \alpha$)  to the number of negative tags($k - k \alpha$), {\sc rel}($T^*$) is the total relevance of tags in $T^*$, $\textsc{rel}_{max}^{T,k}$ is the maximum relevance for $k$ tags from $T$ with the same sentiment in opinion, and {\sc cov}($T^*$) is the total number of item attributes covered by tags in $T^*$. The relevance parameter $\beta$ ensures that the relevance score of tags in $T^*$ is as close to the best possible relevance score $\textsc{rel}_{max}^{T,k}$. The user factor $\alpha$ denotes the proportion of positive and negative tags preferred by a user.

\subsection{General Model}
\label{subsec:datamodel}
We model the TagAdvisor Problem as bipartite graph $G_{TA} = (V = V_T \cup V_I, E)$ as shown in Figure~\ref{fig:TAModel}, where $V_T$ is the set of nodes associated with the tag vocabulary $T$, $V_I$ is the set of nodes associated with the item attribute values, $V_T$ and $V_I$ are disjoint and $E \subseteq (V_T \times V_I)$. The nodes in partite $V_T$ are further classified into positive nodes $V_{T^+}$ (colored green) and negative nodes $V_{T^-}$ (colored red), based on the sentiment of the tags. If the same tag has positive sentiment for an attribute value and negative sentiment for another attribute value, we consider the tag as two different nodes in the set $V_T$. An edge $(t_x^+,a.v_y) \in E$ if $t^+_x$ covers attribute value $a.v_y$, i.e., the rule $\{ \{a.v\}\rightarrow t_x^+ \}$, $a.v_y \in \{a.v\}$ exists; similarly $(t^-_w,a.v_y) \in E$ if $t^-_w$ covers attribute value $a.v_j$. We use the graph model for the coverage purpose.

\vspace{0.05in}
\noindent \textsc{Example} [continued]: \textit{
Figure~\ref{fig:TAModel-example} shows the bipartite graph model of our running example in Table~\ref{tbl:samsungRules}, where $G_{TA}$, has two parts $V_T = V_{T^+} \cup V_{T^-}$ in green and red respectively and $V_I$ in yellow, where $T^+ = \{t^+_1,t^+_2,t^+_3\}$, $T^- = \{t^-_4,t^-_5,t^-_6\}$. The edges represents the rules in Table~\ref{tbl:samsungRules}. For example, nodes $t^+_1$ and $t^-_4$ has three edges to the same attribute value nodes $a.v_4$, $a.v_7$, and $a.v_8$.}

We next define two concrete problem instances of the TA Problem based on {\sc cov}$(T^*)$.

\subsection{Concrete Problem Instances}
\label{subsec:concrete}
In the first problem, {\sc cov}$(T^*)$ is defined as the total number of item attribute values covered by the tags in $T^*$, indepedent of their sentiment. In this problem, an attribute value $a.v_y$ for an attribute $a_y$ of an item $i$ is covered by $T^*$ if $\exists ~t_x \in T^*$ such that $a.v_y \in \textsc{cov}(t_x, i)$, i.e., there exists a tag $t_x$ covering $a.v_y$, independent of its sentiment. 

\vspace{0.05in}
\noindent
\textsc{Definition 2.} Given a set of tags $T^*$, \textsc{Independent-Coverage} of $T^*$ is defined as:
\begin{equation}
\label{eq:UC}
\textsc{cov}_{IC}(T^*) = |\bigcup_{t_x \in T^*} \textsc{cov}(t_x, i)|
\end{equation}  

\vspace{0.05in}
\noindent \textsc{Example} [continued]: \textit{
In the running example in Table~\ref{tbl:samsungRules} and by Figure~\ref{fig:TAModel-example}, if $T^*$ = $\{t_1^+, t_2^+, t_6^-\}$ = \{{\tt super cool, stylish, gimmicky touchscreen}\}, then {\sc cov}$_{IC}(T^*)$ = $|\{a_3, a_4, a_5, a_7, a_8\}|$ = $|\{$\small \sf Color=Red, Front LCD=1.5'', Back LCD= \- 3.5'', Touchscreen=true, Gesture Control=true$\}|$ = 5.} 

Based on $\textsc {cov}_{IC}(T^*)$ in Equation~\ref{eq:UC} the first problem can now be defined as follows.

\vspace{0.05in}
\smallskip\noindent \framebox[\columnwidth]{\parbox{0.9\columnwidth}{ \textsc{Problem 1.} \textbf{[Independent-Coverage TA Problem (IC-TA)]}:
This problem is an instance of TagAdvisor Problem (TA) in Definition 1. where the input and constrains are the same but the objective is:
\begin{itemize}
\item \textsc{cov}$_{IC}(T^*)$ (given by Equation~\ref{eq:UC}) is maximized
\end{itemize}
}}

\vspace{0.05in}
However, by considering the coverage of an item attribute value by a tag independent of the tag's sentiment, we may restrict a user from reviewing both positively and negatively about the different aspects of an item. By {\sc cov}$_{IC}(T^*)$, if $T^*$ includes a tag that is positive and covers a subset of item attribute values, another tag that is negative and covers the same subset would not be included in $T^*$. In the running example in Table~\ref{tbl:samsungRules}, if at least one of the positive tags, say $t_1^+:${\tt stylish} belongs to $T^*$ with a higher relevance score, then $a.v_7$: {\small \sf Touchscreen=true} and $a.v_8$: {\small \sf Gesture~Contro=true}  are considered covered because of rule $\wp: \{a.v_3, a.v_4, a.v_7, a.v_8\}\rightarrow t_1^+$; $T^*$ would not include either of the negative tags {\tt gimmicky touchscreen}  and {\tt poor battery life} related to $a.v_7$, and $a.v_8$. This motivates us to define the second problem instance where an item attribute value is considered {\em fully covered} if it is covered by both positive and negative tags. 

In second problem, coverage of $a.v_y$ depends on the sentiment of its associated tags. An attribute value $a.v_y$ for an attribute $a_y$ of an item $i$ is covered if one of the following holds:

\begin{itemize}[leftmargin=*]
\item $a.v_y$ is covered by both positive and negative tags, and atleast one of its positive and atleast one of its negative tags belong to $T^*$. Formally, $\exists t^+_x \in T^{*}, \exists t^-_w \in T^{-^*} \mbox{ such that } a.v_y \in \mbox{{\sc cov}}(t^+_x, i) \cap a.v_y \in \mbox{{\sc cov}}(t^-_w, i)$
\item $a.v_y$ is covered only by positive tags and not negative tags, and atleast one of its positive tags belongs to $T^*$. Formally, $\exists t^+_x \in T^{*}, \forall t^-_w \in T^{*}$ such that $a.v_y \in \mbox{{\sc cov}}(t^+_x, i) ~\cap~  a.v_y \notin$ \ $~\mbox{{\sc cov}}(t^-_w, i) $
\item $a.v_y$ is covered only by negative tags and not positive tags, and atleast one of its negative tags belongs to $T^*$. Formally, $\forall t^+_x \in T^{+^*}, \exists t^-_w \in T^{-^*} \mbox{ such that } a.v_y \notin \mbox{{\sc cov}}(t^+_x, i) \cap  a.v_y \in ~\mbox{{\sc cov}}(t^-_w, i) $ 
\end{itemize}

\noindent\textsc{ Definition 3.} Given a set of tags $T^*$, \textsc{Dependent-Coverage} of $T^*$ is defined as:
\vspace{-0.1in}
\begin{eqnarray}
\label{eq:DC}
\small
 \textsc {cov}_{DC}(T^*) & = & |(\bigcup_{t^+_x \in T^{*}}\textsc {cov}(t^+_x, i)) \bigcap (\bigcup_{t^-_w \in T^{*}} \textsc {cov}(t^-_w, i))|\nonumber\\
& + & |\bigcup_{t^+_x \in T^{*}} \textsc {cov}(t^+_x, i) \setminus \bigcup_{t^-_w \in T^-} \textsc {cov}(t^-_w, i) |\nonumber\\
& + & |\bigcup_{t^-_w \in T^{*}} \textsc {cov}(t^-_w, i) \setminus \bigcup_{t^+_x \in T^+} \textsc {cov}(t^+_x, i) |
\end{eqnarray}

Thus the coverage function in this problem variant considers both positive and negative tags for an attribute value if it exists; otherwise, it focuses on either the positive tag or the negative tag (which ever exists) and ends up returning the same $T^*$ as Problem 1. In our running example in Table~\ref{tbl:samsungRules}, we see that attribute $a.v_4: {\small \sf Front LCD=1.5''}$ is in three rules corresponding to tags \{{\tt super cool}, {\tt stylish}, and {\tt poor battery life}\}. By this definition of coverage, $a.v_4: {\small \sf Front LCD=1.5''}$ is covered by a tag in $T^*$ if atleast one of the positive tags \{{\tt super cool} or {\tt stylish}\} and the one negative tag {\tt poor battery life} exists in $T^*$. Again, $a.v_3: {\small \sf Color=Red}$ is covered if the positive tag {\tt stylish} belongs to $T^*$ since there is no negative tag related to $a.v_3$ in the rules in Table~\ref{tbl:samsungRules} and $a.v_6: {\small \sf Shutter Speed=8-1/2000}$ is covered if {\tt blurry pictures} is in $T^*$ since there is no positive tag related to $a.v_6$ in the rules in Table~\ref{tbl:samsungRules}. 

\vspace{0.05in}
\noindent \textsc{Example} [continued]: \textit{
In the running example in Table~\ref{tbl:samsungRules} and by Figure~\ref{fig:TAModel-example}, if} $T^*$=$\{t_1^+, t_2^+, t_6^-\}$=\{{\tt super cool, stylish, gimmicky touchscreen}\}, then {\sc cov}$_{DC}(T^*)$ = $|\{a.v_7,a.v_8\}|$ + $|\{a.v_3\}|$=$|\{${\small \sf Touchscreen=true, Gesture Control=true}$\}|$+$|\{${\small \sf Color=\ Red}$\}|$ = 3.

The second problem can now be defined as follows.

\vspace{0.05in}
\smallskip\noindent \framebox[\columnwidth]{\parbox{0.9\columnwidth}{ \textsc{Problem 2.} \textbf{[Dependent-Coverage TA Problem (DC-TA)]}:
This problem is an instance of TagAdvisor Problem (TA) in Definition 1. where the input and constrains are the same but the objective is:
\begin{itemize}
\item \textsc{cov}$_{DC}(T^*)$ (given by Equation~\ref{eq:DC}) is maximized
\end{itemize}
}}

\section{I\lowercase{ndependent}-C\lowercase{overage} T\lowercase{ag}A\lowercase{dvsior} \small{(IC-TA)}}
\label{sec:IC-TA}

In this section, we first analyze the computational complexity of the Independent-Coverage TagAdvsior (IC-TA) problem and show that it is NP-complete; then we discuss exact algorithms and an approximation algorithm for solving it.
\subsection{Computational Complexity}

\noindent The decision version of the IC-TA is defined as follows:

Given a set of rules $\Re = \{\{a.v\}\rightarrow t_x \}$ for an item $i$, non-negative integer budget $k$, relevance parameter $\beta$ ($0 \leq \beta \leq 1$), user factor $\alpha$ ($0 \leq \alpha \leq 1$), and integer threshold $\gamma \geq 0$, is there a set of $T^* \subseteq T$ such that $\textsc {cov}_{IC}(T^*)$ $\geq \gamma$ subject to: $|T^*| \leq k$, {\sc pol}$(T^*) =$ $\frac{\alpha}{1-\alpha}$, and $\textsc{rel}(T^*) \geq \beta \cdot \textsc{rel}_{max}^{T,k}$.

\begin{theorem}
\label{th:comp_Prob2}
The decision version of the Independent-Coverage TagAdvsior (IC-TA) problem is NP-Complete.
\end{theorem}
\begin{proof}
The membership of decision version of IC-TA in NP is obvious. To verify NP-Completeness, we reduce Max-Coverage problem with group budget constraints (MCG)~\cite{Chekuri:2004}, to our problem and argue that a solution to MCG exists, if and only if, a solution to our problem exists. In MCG problem, given $S = \{ S_1, S_2 , ...\}$ as a collection of sets where each set $S_i$ is a subset of a ground set $\mathcal{X}$ of $l$ elements and $S$ is partitioned into groups $G_1, G_2, ..., G_m$, the goal is to pick $k$ sets from $S$ such that at most $k_i$ sets be picked from each group $G_i$ and cardinality of their union is maximum. This problem was proved to be NP-Complete by reduction from Max-Coverage in~\cite{Chekuri:2004} if the number of groups is atleast one ($m \geq 1$). We construct an instance of IC-TA problem such that the solution for MCG with two groups exists, if and only if, the solution to our IC-TA instance exists.

For every $S_i \in S$, there exists a corresponding $t_x \in T$. We create a set of rules $\Re =\{ \{a.v\}\rightarrow t_x \}$ such that for every element in ground set $\mathcal{X}$, there exist a tag  such that $a.v_i \in t_x$.
Next, based on the sentiment of the tags, we partition $\Re$ into two groups, i.e., positive and negative groups where  $G_1$ corresponds to positive group and $G_2$ corresponds to negative group. We set the $\alpha=\frac{k_1}{k_2}$ , where $k_i$ is number of sets should be picked from each group $G_i$, and $\beta = 0$ i.e., the polarity constraint {\sc pol}$(T^*) \geq$ $\frac{\alpha}{1-\alpha}$ is satisfied and relevance constraint will be relaxed because $\textsc{rel}(T^*) \geq 0$ is always true.
In Equation~\ref{eq:UC}, $\textsc {cov}_{IC}(T^*)$ is the cardinality of the union of the coverage of the tags. Thus, in this IC-TA instance, if $T^*$ with $k = k_1 + k_2$ tags, where $k_1$ tags are selected from positive group and $k_2$ tags are selected from negative group maximizes the $\textsc {cov}_{IC}(T^*)$, then the corresponding sets in $S$ maximizes the cardinality of their union in MCG with two groups. Thus, IC-TA problem is NP-Complete.
\end{proof}

\subsection{Exact Algorithms}
\label{subsec:exact-IC}
A brute-force approach to solve the IC-TA problem enumerates all possible $^nC_k$ ($n$ is the total number of tags in vocabulary, $k$ is the size of $T^*$) combinations of tags in order to return the optimal set of tags maximizing coverage \textsc{cov}$_{IC}(T^*)$ and satisfying the constraints. The number of possible candidate sets is exponential in the number of the rules for an item. If there are $m$ boolean attributes for an item, there are potentially $2^m$ rules for tags. Thus, evaluating the constraints on each of the candidate sets and selecting the optimal result can be prohibitively expensive. 
We refer to this naive exact algorithm of IC-TA as {\bf E-IC-TA}. 

We next show how IC-TA problem can be described in an Integer Linear Programming (ILP) framework. We refer to it as {\bf ILP-IC-TA}. Let $\{x^+_1,x^+_2,...\}$ be integer variables such that if $t^+_i \in T^*$ then $x^+_i = 1$, else $x^+_i = 0$. Similarly, $\{x^-_1,x^-_2,...\}$ are integer variables such that if $t^-_i \in T^*$ then $x^-_i = 1$, else $x^-_i = 0$. Let $\{y_1,y_2,...\}$ be integer variables such that  $y_j = 1$ if $a.v_j$ is covered by either positive or negative tag. The ILP version of IC-TA problem is given by Equation~\ref{equ:ILP-IC}.
\begin{equation*}
\label{equ:ILP-IC}
\begin{aligned}
& \underset{}{\text{Maximize}}
& & \mathrm \sum_{a.v_j} y_j  \\ 
& \text{subject to}
& & \sum x^+_i + \sum x^-_i \leq k \\ 
&&& \frac{\sum x^+_i}{\sum x^-_i} = \frac{\alpha}{1-\alpha}\\
&&& \textsc{rel}(T^*) \geq \beta \cdot \textsc{rel}_{max}^{T,k} \\ 
&&& \sum_{a.v_j \in t^+_i} x^+_i + \sum_{a.v_j \in t^-_i} x^-_i \geq y_j \\ 
&&& y_j \in \{0,1\} ~(if~ y_j=1 ~then ~a.v_j ~is ~covered)\\
&&& x^+_i \in \{0,1\} ~(if ~x^+_i=1 ~then ~t^+_i ~is ~selected)\\
&&& x^-_i \in \{0,1\} ~(if ~x^-_i=1 ~then ~t^-_i i~s ~selected)\\
\end{aligned}
\end{equation*}

The first three constraints are related to the size of the $T^*$, polarity, and relevance and the last constraint shows that $a.v_j$ is covered if at least one tag (positive or negative) which are dependent to $a.v_j$ are selected. Note that the ILP-IC-TA only works well for moderate-sized problem. We next develop a practical algorithm to solve IC-TA problem efficiently.

\vspace{-0.10in}
\subsection{Approximation Algorithm (A-IC-TA)}
\label{subsec:approx-IC}
In order to solve IC-TA problem, we consider the Max-Coverage problem with group budget constraints (MCG) problem variant in Chekuri et al.'s paper~\cite{Chekuri:2004}, where given $S = \{ S_1, S_2 , ...\}$ as a collection of sets where each set $S_i$ is a subset of a ground set $\mathcal{X}$ and $S$ is partitioned into groups $G_1, G_2, ..., G_m$, the goal is to pick $k$ sets from $S$ such that at most $k_i$ be picked from each group $G_i$ and cardinality of their union is maximum. The authors in~\cite{Chekuri:2004} proposed a greedy solution with a 2-approximation algorithm.

In our problem, the set $S$ is the set of rules $\Re =\{ \{a.v\}\rightarrow t_x \}$ which is partitioned into two groups based on the tags sentiments. We use the similar greedy approach in~\cite{Chekuri:2004} and we check an extra constraint for the relevance. Intuitively, the greedy approach will iteratively picks those relevant tags that cover the maximum number of uncovered item attribute values.

Algorithm~\ref{algo:Greedy2} is the pseudo code for our algorithm, denoted as \textbf{A-IC-TA}. The A-IC-TA algorithm iteratively picks tags from $T$ that cover the maximum number of uncovered item attribute values such that the number of positive and negative tags are $k_1 = \lceil \alpha k \rceil$, $k_2 = k-k_1$ and $\textsc{rel}(T^*) \geq \beta \cdot \textsc{rel}_{max}^{T,k}$. If we assume all tags in $T^+$ and $T^-$ are sorted by their relevance, the $\textsc{rel}_{max}^{T,k}$ is the summation of the first $k_1$ positive tags and $k_2$ negative tags in the sorted list . More specifically, let us assume a positive tag $t_y$ is picked. At step $x$, where there are $x-1$ tags in $T^*$, Algorithm~\ref{algo:Greedy2} iteratively adding one tag with highest coverage to $T^*$, where its relevance score is atleast $\beta \cdot \textsc{rel}_{max}^{T,x}$.

\noindent \textsc{Example} [continued]: \textit{
In the running example, for $k=2$, $\alpha = 0.5$, and $\beta = 0.5$, Algorithm~\ref{algo:Greedy2} returns $T^*$ = \{{\tt stylish}, {\tt blurry pictures}\}. In first iteration, the highest relevance score of the positive tags $\textsc{rel}_{max}^{T,1}$ is $0.3$. Among the positive tags {\tt super~cool} and {\tt stylish} has relevance larger than $0.15 = 0.5 \cdot 0.3$ and coverage score $3$ and $4$. Thus  {\tt stylish} with highest coverage score of $4$ will be selected. Next, the highest relevance score of the negative tags is $0.15$, among all the negative tags whose relevance are larger than $0.075 = 0.5 \cdot 0.15$,  {\tt blurry pictures} with highest coverage of $3$ will be selected.}

\begin{theorem}
\label{th:IC-approxRatio}
The A-IC-TA Algorithm provides near optimal solution with $2$-approximation factor.
\end{theorem}
\begin{proof}
The proof follows from the $2$-approximation factor proof of the algorithm for solving the Max-Coverage with group budget constraints (MCG) problem in~\cite{Chekuri:2004} with additional constraint over the relevance. We are given an integer $k$, and an integer bound $k_1$  and $k_2$ for two sets $T^{+^*}$ and $T^{-^*}$  i.e., positive and negative tags. A solution is a subset $T^* \subseteq T$ such that $|T^*| \leq k$, $|T^* \cap  T^{+^*}| \leq k_1$, $|T^* \cap  T^{-^*}| \leq k_2$,  and $\textsc{rel}(T^*) \geq \beta \cdot \textsc{rel}_{max}^{T,k}$. The objective is to find the solution such that the number of item attribute values covered by $T^*$ is maximized. Without loss of generality we assume that $k_1$, $k_2$ is equal to one, otherwise we make a copies of each set $T^{+^*}$ and $T^{-^*}$. 

In jth iteration, let $t^g_j$ be the tag that greedy Algorithm~\ref{algo:Greedy2} (A-IC-TA ) picks and let $t^o_j$ be the tag that OPT picks. We let $C'_j = t^g_j \setminus \cup^{j-1}_{h=1} t^g_h$ denote the set of new item attribute values that A-IC-TA adds in jth iteration. Let $\textsc {cov}_{IC}(T^*_{gr}) = |\cup_j t^g_j|$ and $\textsc {cov}_{IC}(T^*_{op}) = |\cup_j t^o_j|$ denote the coverage of the A-IC-TA and optimal solution.

We first show that for $1 \leq j \leq k$, $|C'_j| \geq |t^o_j \setminus T^*_{gr}|$. Obviously when $t^o_j \setminus T^*_{gr} = \emptyset$, it holds. When the greedy algorithm A-IC-TA picked tag $t^g_j$, the set $t^o_j$ was also available and the relevance constraint should have satisfied but greedy didn't picked it because $|C'_j|$ was atleast $| t^o_j - \cup^{j-1}_{h=1} t^g_h |$. Since $\cup^{j-1}_{h=1} t^g_h \subseteq \textsc T^*_{gr}$, $|C'_j|$ is atleast $|t^o_j \setminus T^*_{gr}|$.
\begin{eqnarray*}
\textsc {cov}_{IC}(T^*_{gr}) & = & \sum_{j} |C'_j| \\
&\geq &\sum_j|t^o_j \setminus T^*_{gr}|\\
&\geq & |\cup_j t^o_j| - \textsc {cov}_{IC}(T^*_{gr}) \\
&\geq & \textsc {cov}_{IC}(T^*_{opt}) - \textsc {cov}_{IC}(T^*_{gr})\\
\end{eqnarray*}  
 Thus $\textsc {cov}_{IC}(T^*_{gr}) \geq \frac{1}{2} \textsc {cov}_{IC}(T^*_{opt})$.
\end{proof}

\begin{algorithm}[t]
\SetAlgoNoLine
  	\SetKwInOut{Input}{Input}
  	\SetKwInOut{Output}{Output}
  	\SetKwFunction{Compute}{Compute}
	\Indm
		\Input{Tag vocabulary $T$, set of rules $\Re = \{\{a.v\}\rightarrow t_x \}$, budget $k > 0$, relevance parameter $0 < \beta \leq 1$, user factor $0 < \alpha \leq 1$}
		\Output{set of tags $T^* \subseteq T$ of size $k$}
	\Indp
	$k_1 = \lceil k \cdot \alpha \rceil$; $k_2 = k - k_1$\;
	$T^* = \emptyset$\;
	\For{$x = 1$ \KwTo $k$}
	{ 
		\For{$t_y \in T \setminus T^*$}{ 
			\If {$(t_y \in T^+$ and $|T^{+^*}|<k_1)$ or $(t_y \in T^-$ and $|T^{-^*}|<k_2)$}{
					{\bf if} {$\textsc{rel}(T^* \cup t_y) \geq \beta \cdot \textsc{rel}_{max}^{T,x}$}
						{\bf  then} \Compute{$\textsc{cov}_{IC}(T^* \cup t_y)$}\;	
			}
		}
		$t_y =  \argmaxl_{t_y \in T \setminus T^*}\textsc{cov}_{IC}(T^* \cup t_y)$\;
		$T^* = T^* \cup t_y$\;
	}
	\Return{$T^*$}
	\caption{IC-TA Algorithm (A-IC-TA)}
	\label{algo:Greedy2}
\end{algorithm}

\section{D\lowercase{ependent}-C\lowercase{overage} T\lowercase{ag}A\lowercase{dvsior} (DC-TA)}
\label{sec:DC-TA}

In this section, we focus on the Dependent-Coverage TagAdvsior (DC-TA) problem. We first propose a graph model for the problem, then analyze its computational complexity and prove that it is NP-complete, and finally develop an exact algorithm and an efficient constant factor approximation algorithm for solving it. 

In order to solve the DC-TA problem, we transform the bipartite graph in Figure~\ref{fig:TAModel} to a labeled graph $G_{DC-TA} = (V_T,E)$, where $V_T$ is the set of nodes associated with the tag vocabulary $T$, and $E \subseteq (V_T \times V_T)$. Each edge $e \in E$ has a label, $l: E \rightarrow \{a.v_i\}$. We define an edge label $l(v_{t_{x_1}}, v_{t_{x_2}})$ as the dissimilarity between two tag nodes $v_{t_{x_1}}, v_{t_{x_2}} \in V_T$. We can consider each tag as a boolean vector of size $m$ (number of item attributes) where bit at location $y$ is $1$ if $a.v_y \in$ {\sc cov}($t_x, i$). Using such a vector representation of the tags, we define label of an edge $(v_{t_{x_1}}, v_{t_{x_2}})$ as a set of all different item attribute values. In other words$l(v_{t_{x_1}}, v_{t_{x_2}}) = \{a.v_i\}$, where $i$ is a bit location which are different. In previous proposed model in~\cite{Nazi:2015}, we have used the hamming distance between the vector representation as the edge weight. However, using the absolute distance would not be enough for DC-TA problem. In the new proposed model the dissimilarity between multiple tags would be the cardinality of the union of the edge labels.
In our framework, $T$ is partitioned into two disjoint sets: $T^+$ and $T^-$ based on tag sentiment. Thus there can be three kind of node-to-node connectivity: $v_{t^{+}_{x_1}}$ ($t^+_{x_1} \in T^+$) is connected to $v_{t^{+}_{x_2}}$ ($t^+_{x_2} \in T^+$), $v_{t^{-}_{w_1}}$ ($t^-_{w_1} \in T^-$) is connected to $v_{t^{-}_{w_2}}$ ($t^-_{w_2} \in T^-$), and  $v_{t^{+}_{x_1}}$ ($t^+_{x_1} \in T^+$) is connected to $v_{t^{-}_{w_2}}$ ($t^-_{w_2} \in T^-$). 
The first two connectivities are intra-edges and the third belongs to the category of cross-edges.

\begin{figure*}
\centerline{
\includegraphics[height=56mm, width = 165mm]{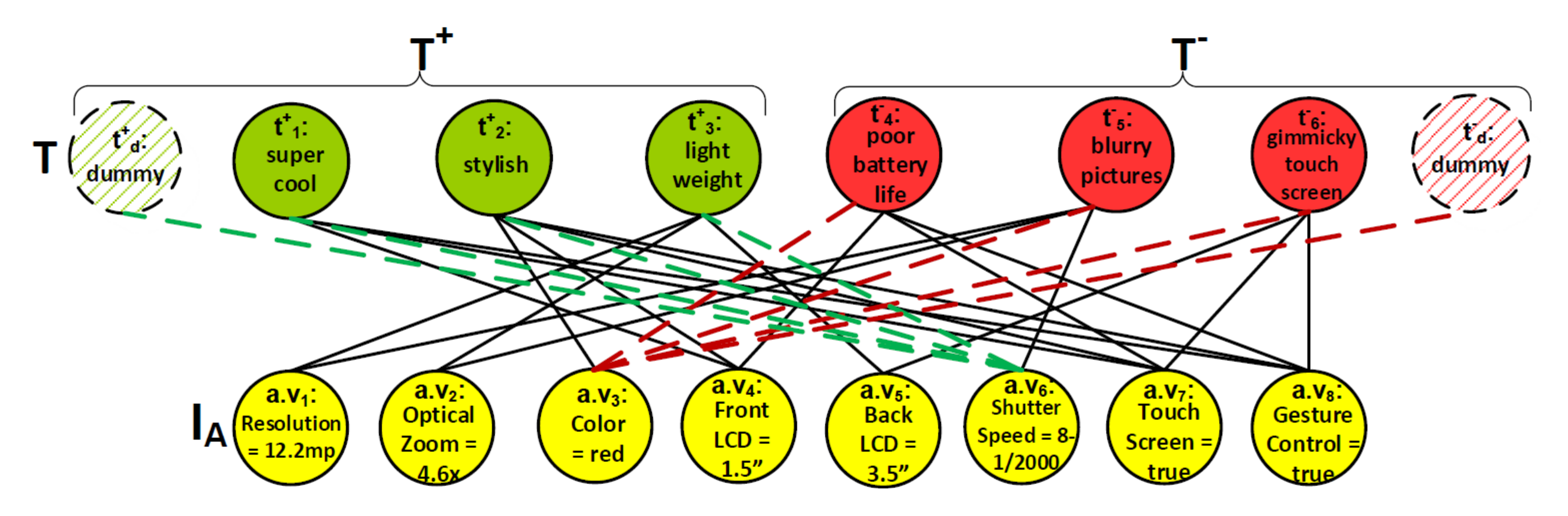}}
\vspace{-0.2in}
\caption{{TA Graph model of the Running Example with dummy nodes and edges.}}
\label{fig:TAModel-dummy-example}
\end{figure*}

Recall that the coverage function $\textsc {cov}_{DC}(T^*)$ discussed in Equation~\ref{eq:DC} is based on three conditions that considers both positive and negative tags for an attribute value if it exists; otherwise, it focuses on either the positive tag or the negative tag. We argue that we can reduce the last two conditions to the first one by introducing {\em dummy nodes and edges}. In other words, for attribute values with only positive tags, selecting any negative tag would not influence their coverage; hence we can add a dummy negative tag $t_d^-$ and add dummy edges from those attribute value nodes to all the negative tags. Similarly, for attribute values with only negative tags, selecting any positive tag would not influence their coverage and we can add dummy positive tag $t_d^+$ and add dummy edges from those attribute value nodes to all the positive tags.

Figure~\ref{fig:TAModel-dummy-example} shows the original bipartite graph in Figure~\ref{fig:TAModel} with dummy nodes and edges for the running example in Table~\ref{tbl:samsungRules}. Since node {\sf \small Color=Red} is not covered by any of the negative tag nodes \{{\tt poor battery life}, {\tt blurry pictures}, {\tt gimmicky touchscreen}\}, we add a dummy negative tag $t_d^-$  (red shaded area) and dummy edges (red dotted lines) from node {\sf \small Color=Red} to all the negative tags. Similarly, the dummy positive tag $t_d^+$ (green shaded area) is added and dummy edges (green dotted lines) are added from {\sf \small Shutter~Speed=8-1/2000} to all positive tag nodes \{{\tt super cool}, {\tt stylish}, {\tt lightweight}, $t_d^+$\}. Figure~\ref{fig:DC-TAModel} shows the graph $G_{DC-TA}$ of our running example in Table~\ref{tbl:samsungRules} having $8$ nodes $T$= $\{t^+_1$, $t^+_2$, $t^+_3$, $t_d^+$, $t^-_4$, $t^-_5$, $t^-_6$, $t_d^-,\}$ = \{{\tt super cool}, {\tt stylish}, {\tt lightwe\-ight}, $t_d^+$, {\tt dummy positive}, {\tt poor battery life}, {\tt blurry pictures}, {\tt gimm\-icky touchscreen}, {\tt dummy negative}, $t_d^-$\}; the label of an edge $(t_i,t_j)$ shows the item attribute values which are not covered by $t_i$ and $t_j$, i.e, in vector representation of the tags, those bits which are different. For example, the edge label between the $t^+_1$:{\tt super~cool} and $t^-_4$:{\tt poor battery life} is $l(t^+_1, t^-_4) = \{ a_3, a_6\}$. By Figure~\ref{fig:TAModel-dummy-example}, $t^+_1$:{\tt super~cool} is connected to $a.v_4$:{\sf \small Fro\-nt LCD=1.5''}, $a.v_6$:{\sf \small Shutter Speed=8-1/2000}, $a.v_7$:{\sf \small Touchscreen=true}, $a.v_8$:{\sf \small Gesture Control=true}. The vector representation of $t^+_1$ is $[0,0,0,$ $1,0,1,1,1]$. Similarly, $t^-_4$ can be represented as $[0,0,1,1,0,0,1,1]$, i.e, they are different in $a_3$ and $a_6$. Note that the size of the edge label show the dissimilarity between two tags measured by Hamming metric. The Hamming distance between $t^+_1$ and $t^-_4$ is $|l(t^+_1, t^-_4)| = |\{ a_3, a_6\}| = 2$. 

\begin{figure}[t]
\vspace{-0.1in}
\centerline{
\includegraphics[height=84mm,width = 93mm]{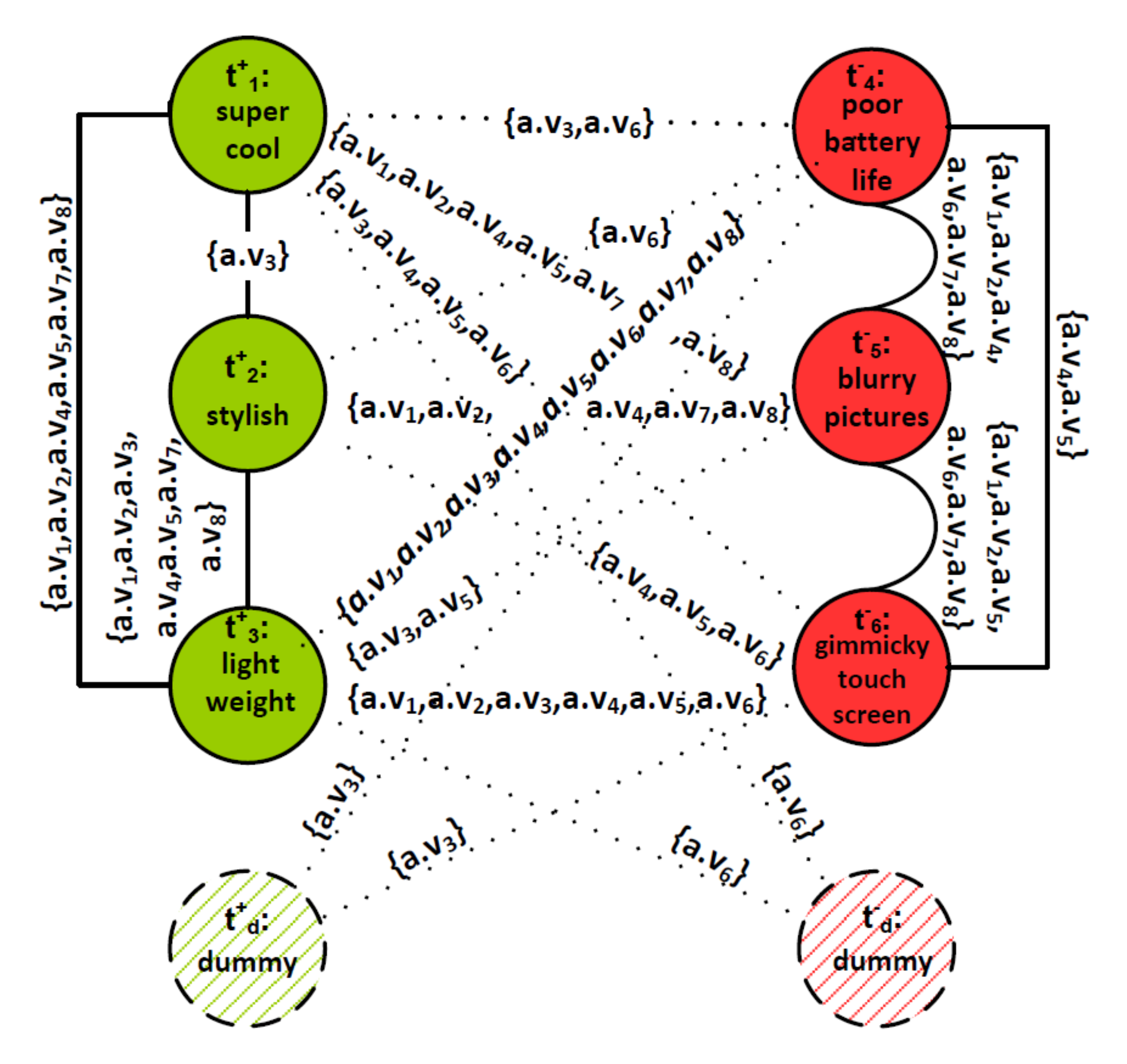}}
\vspace{-0.2in}
\caption{{DC-TA Graph model of Running Example}}
\label{fig:DC-TAModel}
\end{figure}

Our objective in this problem is to maximize $\textsc {cov}_{DC}(T^*)$. Considering this transformed labeled graph model, the goal is to minimize the number of item attribute values which are not covered. We would select positive tags  $T^{+^*}$ and negative tags $T^{-^*}$ from nodes in $V_{T^+}$ and $V_{T^-}$ respectively such that the constraints are satisfied and the size of the union of the labels of the cross-edges minus the union of the labels of the intra-edges is minimum in the induced graph. 
Formally, the objective of DC-TA in this graph model is to minimize:
\begin{equation}
\label{equ:core_pol_and}
\vartheta_{DC}(T^*) = |\bigcup_{\substack{t_x \in \{T^{+^*} \cup t_d^+\}\\t_w \in \{T^{-^*}\cup t_d^-\} }} l(v_{t_x},v_{t_w}) \setminus \bigcup_{\substack{t_x,t_w \in T^{+^*}\\t_x,t_w \in T^{-^*}}} l(v_{t_x},v_{t_w})|
\end{equation}

Where the first term is union of the labels of the cross-edges (edges between positive-negative tags) and the second term is the union of the labels of  the intra-edges (edges between positive-positive and negative-negative tags). We can observe that minimizing $\vartheta_{DC}(T^*)$ is equivalent to maximizing the $\textsc {cov}_{DC}(T^*)$. 
$\textsc {cov}_{DC}(T^*)$ is based on the three different conditions over the item attribute values. Due to the inclusion of dummy nodes and edges, the problem reduces to one condition which is maximizing the similarity of positive and negative tags. Clearly, minimizing the positive and negative tags dissimilarity by Equation~\ref{equ:core_pol_and} is equivalent to maximizing the similarity of those tags. Next, we analyze the computational complexity of this problem.

\subsection{Computational Complexity}

\noindent
The decision version of the DC-TA is defined as follows:

Given graph $G_{DC-TA} = (V_T,E)$, non-negative integer budget $k$, relevance parameter $\beta$ ($0 \leq \beta \leq 1$), user factor $\alpha$ ($0 \leq \alpha \leq 1$), and integer threshold $\gamma \geq 0$, is there a set of $T^* \subseteq T$ such that $\vartheta_{DC} \leq \gamma$ subject to: $|T^*| \leq k$, {\sc pol}$(T^*)$ = $\frac{\alpha}{1-\alpha}$ and $\textsc{rel}(T^*) \geq \beta \cdot \textsc{rel}_{max}^{T,k}$.
\begin{theorem}
\label{th:comp_Prob3}
The decision version of the Dependent-Coverage TagAdvsior (DC-TA) problem is NP-Complete.
\end{theorem}
\begin{proof}
It is obvious that the decision version of the DC-TA is in NP. To verify NP-Completeness, we reduce the MAX-SUM Facility Dispersion problem~\cite{Ravi:1994,Erkut:1990,Hansen:1988} to our problem and argue that a solution to MAX-SUM Facility Dispersion exists, if and only if, a solution to our problem exists. In MAX-SUM Facility Dispersion problem, given a set of $V = \{v_1, v_2, ..., v_n\}$ of $n$ nodes, a non-negative distance $w(v_i,v_j)$ for each pair of nodes $v_i$, $v_j$, and an integer $p$ smaller than $n$, the goal is to find a subset $P= \{v_{i_1},v_{i_2}, ..., v_{i_p}\}$ of $V$, with $|P|=p$, such that sum of distances are maximized. This problem was proved to be NP-Complete~\cite{Erkut:1990,Hansen:1988}. We construct an instance of DC-TA problem such that the solution for MAX-SUM Facility Dispersion exists, if and only if, the solution to our DC-TA instance exists.

We create a graph $G_{DC-TA} = (V_T,E)$ such that for every $v_i \in V$ there is a corresponding node $v_{t_{x_i}} \in V_T$ and a distance $w(v_i,v_j)$ corresponds to the Hamming distance of two tags $t_{x_1}$ and $t_{x_2}$, i.e, $|l(v_{t_{x_1}}, v_{t_{x_2}})|$. Let in this DC-TA instance, $\alpha = 1$, i.e., $k_1 = p$, and $k_2 = 0$ (only $p$ positive tags should be selected). Also by setting $\beta = 0$ the relevance constraint will be relaxed because $\textsc{rel}(T^*) \geq 0$ is always true. Let in DC-TA instance, positive and negative tags cover exactly same item attribute values, i.e., the label of all cross-edges is an empty set, i.e., distance between positive and negative tags is $0$. In DC-TA instance, assume label of the edges among positive tags are disjoint, i.e $|\cup_{\substack{t_x,t_w \in T^{+^*}}} l(v_{t_x},v_{t_y})| $ is equal to the sum of the hamming distance. Thus, the DC-TA problem collapses to that of finding $p$ positive tags such that $-\sum_{\substack{t_x,t_w \in T^{+^*}}} |l(v_{t_x},v_{t_y})|$ is minimum or sum of the hamming distances is maximum.
Thus, in this DC-TA instance, if $T^*$ with $p$ positive tags and zero negative tags maximizes the $\textsc {cov}_{DC}(T^*)$, then the corresponding nodes in $V$ maximizes the sum of distances in MAX-SUM Facility Dispersion. Thus, DC-TA problem is NP-Complete.
\end{proof}

\subsection{Exact Algorithms}
\label{subsec:exact-DC}
Similar to Section~\ref{subsec:exact-IC}, a brute-force approach to solve the DC-TA problem enumerates all possible $^nC_k$ combinations of tags in order to return the optimal set maximizing coverage \textsc{cov}$_{DC}(T^*)$ (or, minimizing $\vartheta_{DC}(T^*)$) and satisfying the constraints. We refer to this computationally prohibitive exact algorithm of DC-TA as {\bf E-DC-TA}. We next show how DC-TA problem can be described in an Integer Linear Programming (ILP) framework. We refer to it as {\bf ILP-DC-TA}. Let $\{x^+_1,x^+_2,...\}$ be integer variables such that if $t^+_i \in T^*$ then $x^+_i = 1$, else $x^+_i = 0$. Similarly, $\{x^-_1,x^-_2,...\}$ is integer variables such that if $t^-_i \in T^*$ then $x^-_i = 1$, else $x^-_i = 0$. Let $\{y_1,y_2,...\}$ be integer variables. Since an item attribute is covered if both positive and negative tags are selected so when $a_j$ is covered then $y_j = 2$. The ILP version of DC-TA problem is given by Equation~\ref{equ:ILP-DC}.

\vspace{-0.1in}
\begin{equation*}
\label{equ:ILP-DC}
\begin{aligned}
& \underset{}{\text{Maximize}}
& & \mathrm \sum_{a.v_j} y_j \\ 
& \text{subject to}
& & \sum x^+_i + \sum x^-_i \leq k+2 \\ 
&&& \frac{\sum_{t_i \in T^+} x^+_i}{\sum_{t_i \in T^-} x^-_i} = \frac{\alpha}{1-\alpha}\\
&&& \textsc{rel}(T^*) \geq \beta \cdot \textsc{rel}_{max}^{T,k} \\ 
&&& \sum_{a.v_j \in t^+_i \cap t^-_i} x^+_i + x^-_i \geq y_j \\
&&& y_j \in \{0,2\} ~(if~ y_j=2 ~then ~a.v_j ~is ~covered )\\
&&& x^+_d = 1 , x^-_d = 1 ~(dummy ~tags ~are ~selected)\\
&&& x^+_i \in \{0,1\} ~(if ~x^+_i=1 ~then ~t^+_i ~is ~selected)\\
&&& x^-_i \in \{0,1\} ~(if ~x^-_i=1 ~then ~t^-_i ~is ~selected)\\
\end{aligned}
\end{equation*}

\noindent Recall that dummy nodes and edges are added to reduce the last two conditions of the Equation~\ref{eq:DC}. We add the two dummy tags $t^+_d$, $t^+_d$ to the result set, i.e. $x^+_d $ and $x^-_d$ are set to one. Thus size of the $T^*$ is increased by 2. The next two constraints are related to the polarity, and relevance and the last constraint shows that $a_j$ is covered if both dependent positive or negative tags are selected. We next develop an efficient algorithm to solve IC-TA problem.

\subsection{Approximation Algorithm (A-DC-TA)}
\label{subsec:approx-DC}

Given graph $G_{DC-TA} = (V_T,E)$ as DC-TA model, relevance parameter $\beta$, and user factor $\alpha$, the goal is to select $k_1 = \lceil\alpha k\rceil$ positive tags and $k_2 = k-k_1$ negative tags such that $\textsc{rel}(T^*) \geq \beta \cdot \textsc{rel}_{max}^{T,k}$ and $\vartheta_{DC}(T^*)$ is minimum.

First, we show that such $\vartheta_{DC}(T^*)$ is not submodular. In submodular functions the incremental gain of adding an element to a set decreases as the size of the set increases, i.e., in the context of our paper, for all tags $t_x$  and $S \subseteq T$, $F(S\cup\{t_x\})-F(S) \geq F(T\cup{t_x})-F(T)$. The authors in~\cite{Nemhauser:1978} proved that if a function is monotone and submodular, the greedy approach provides near optimal solution with $(1-1/e)$-approximation factor. We prove that $\vartheta_{DC}(T^*)$ is not submodular, thus, there is not any greedy approach provides near optimal solution with $(1-1/e)$-approximation factor for DC-TA problem. Next we propose an approximation algorithm, denoted by \textbf{A-DC-TA} and we prove its approximation factor.

\vspace{-.05in}
\begin{theorem}
\label{th:prob3_Submodularity}
The function $\vartheta_{DC}(T^*)$ is not submodular.
\end{theorem}
\begin{proof}
Let $T_1 = T^+_1 \cup T^-_1$ be the set of positive and negative tags for item $i$ covering item attribute values $\{a.v\}$ with $I_{A1} \subseteq I_A$ as schema. Let $T_2 \subseteq T_1$ covers attribute values $\{a.v\}$ with schema $I_{A2} \subseteq I_{A1}$, such that $T_2$ has the same positive tags $T^+_2 = T^+_1$ but $T^-_1$ has more negative tags than the $T^-_2$, i.e., in $T_1$ there are some values for attributes  $\{a\}$ that are cover by negative tags, $\{a\}\subseteq I_{A1}$, which those attribute values are not covered by $T_2$, $\{a\} \nsubseteq I_{A2}$. Now assume we want to add to both sets a positive tag $t_x^+$ that covers some values of attributes $\{a'\} \subseteq \{a\}$. In DC-TA problem every attribute values associated with both positive and negative tags is covered if atleast one from each negative and positive tags are selected. It is clear that adding $t_x^+$ to $T_1$ is more beneficial than adding it to $T_2$ because all values of attributes $a_j \in \{a'\}$ are covered by $T_1$ by both positive and negative tags but they are only covered by $T_2$ by positive tag but not negative. Thus, the incremental gain of adding this tag to a set increased as the size of the set increases, which contradicts with submodularity, where the incremental gain of adding a tag to a set should decreases as the size of the set increases.
\end{proof}

We develop a greedy algorithm~\ref{algo:Greedy3} and theoretically prove that it produces a solution with constant factor approximation of the optimal. The A-DC-TA Algorithm uses the user factor $\alpha$ to find the number of positive and negative tags need to be selected, i.e. $k_1$ and $k_2$. Let $t_x \in T^+ \setminus T^*$ and $t_y \in T^-\setminus T^*$ be the tags with highest relevance score in positive and negative tags which have not been selected yet. Lines $3-12$ of the algorithm iteratively picks the cross-edges $(v_{t_x},v_{t_y}) ,t_x \in T^+, t_y \in T^-$ with the relevance score of atleast $\beta \cdot \textsc{rel}_{max}^{T,k'}$, which add minimum weight to $\vartheta_{DC}(T^*)$ and adds those tags to the $T^*$ until the number of selected positive or negative tags be $k_1$ or $k_2$. If the number of selected positive and negative tags is $k_1$ and $k_2$, the algorithm returns $T^*$ as the top-$k$ tags, otherwise there are still more tags that should be selected from either positive or negative tags (not both). Let us assume $k_2$ negative tags are selected. The algorithm (line $14-22$) finds the new tag $ t_y \in T^+ \setminus T^*$ with the relevance score of atleast $\beta \cdot \textsc{rel}_{max}^{T,k'}$, which add minimum weight to $\vartheta_{DC}(T^*)$. Similarly, if all $k_1$ positive tags are selected but still negative tags are less that $k_2$ then new tag $ t_y \in T^- \setminus T^*$ will be selected (line $23-30$).

\vspace{0.05in}
\noindent \textsc{Example} [continued]: \textit{
In the running example, for $k=2$, $\alpha = 0.5$, and $\beta = 0.5$, solving the problem with practical heuristic Algorithm~\ref{algo:Greedy3} returns $T^*$ = \{{\tt stylish}, {\tt poor battery life}\}. It first finds $t^+_1$ = {\tt super cool} and $t^-_6 $= {\tt gimmicky touchscreen} as the positive and negative tags with highest relevance scores $0.3$ and $0.15$  ($\textsc{rel}_{max}^{T,2}= 0.45$). Then it selects $t^+_2$= {\tt stylish} and $t^-_4$ = {\tt poor~battery~life} because $\vartheta_{DC}(\{t^+_2, t^-_4\}) = 1$ gives the smallest value among other selections and it satisfies the relevance constraint, i.e. \textsc{rel}($\{t_x,t_y\}$) $\geq \beta \cdot \textsc{rel}_{max}^{T,2}$ (their relevance $0.33 = 0.2+0.13$ is greater that $0.225= 0.5 \cdot 0.45$)}

\begin{theorem}
\label{th:core_pol_and_approx}
The proposed heuristic DC-TA algorithm~\ref{algo:Greedy3} produce a solution with $2$-approximation of the optimal, i.e. $\vartheta_{DC}(T^*_{gr})$ $\leq 2\cdot\vartheta_{DC}(T^*_{opt})$.
\end{theorem}

\begin{proof}
Algorithm~\ref{algo:Greedy3} picks an edge in each iteration. let us assume in jth iteration, $e_j$ and $e'_j$ be an edge selected by greedy and optimal respectively. $C'_{e_j}$ denotes the set of item attribute values that are not covered in first jth iterations, i.e., $C'_{e_j} =  \cup^{j}_{h=1} C'_{e_h}$. Thus, the number of item attribute values which are not covered by the A-DC-TA would be $\vartheta_{DC}(T^*_{gr_k})  = |\cup_{j} C'_{e_j}| $. Similarly $\vartheta_{DC}(T^*_{opt_k})  = |\cup_{j} C'_{e'_j}|$ shows the number of item attribute values which are not covered by the optimal algorithm. 

Let us assume at step j optimal algorithm picks $e'_j$ but the greedy algorithm picks $e_j$. The reason that greedy algorithm didn't pick the $e'_j$ is that the number of item attribute values that are not covered in j iterations by selecting $e_j$ is less than the the number of item attribute values that are not covered by selecting $e'_j$, i.e., $|C'_{e_j}| \leq |C'_{e'_j} \bigcup \cup^{j-1}_{h=1} C'_{e_h}|$. Thus, $|C'_{e_j}| \leq |C'_{e'_j}| + |\cup^{j-1}_{h=1} C'_{e_h}|$. Since $|\cup^{j-1}_{h=1} C'_{e_h}|$ is at least $|\cup^{k}_{h=1} C'_{e'_h}|$, we have $|C'_{e_j}| \leq |C'_{e'_j} |+ |\cup^{k}_{h=1} C'_{e'_h}|$. Using this inequality we have:

\begin{eqnarray*}
\vspace{-.1in}
\vartheta_{DC}(T^*_{gr_k})  & = & |\cup_{j} C'_{e_j}| \\
&\leq &|\cup_{j} C'_{e'_j}| + |\cup^{k}_{h=1} C'_{e_h}|\\
&\leq &   \vartheta_{DC}(T^*_{opt_k}) + \vartheta_{DC}(T^*_{opt_k})\\
&\leq & 2\vartheta_{DC}(T^*_{opt_k})
\end{eqnarray*}  
Thus the A-DC-TA produces a solution with $2$-approximation of the optimal.
\end{proof}

\begin{algorithm}[h!]
\SetAlgoNoLine
  	\SetKwInOut{Input}{Input}
  	\SetKwInOut{Output}{Output}
	\Indm
		\Input{$G_{DC-TA} = (V_T,E)$, budget $k > 0$, user factor 0 $< \alpha \leq 1$, relevance importance $0 < \beta \leq 1$}
		\Output{set of tags $T^* \subseteq T$ of size $k$}
	\Indp
	$k_1 = \lceil k \alpha \rceil$; $k_2 = k - k_1$\;
	$T^* = \emptyset$\;
	\While{$(k_1 > 0$ and $k_2 > 0)$}{
		$k' = | T^* | + 2$ \;
		\For{$e=(t_x,t_y), (t_x\in T^+ \setminus T^* ,t_y \in T^- \setminus T^*)$}{
		\textbf{if} {\textsc{rel}($T^* \cup \{t_x,t_y\}$) $\geq \beta \cdot \textsc{rel}_{max}^{T,k'}$}{
						\textbf{then} \Compute{$\vartheta_{DC}(T^* \cup \{t_x,t_y\})$}\;}
					}
					$T^* = T^* \cup \argminl_{t_x\in T^+ \setminus T^* ,t_y \in T^- \setminus T^*} \vartheta_{DC}(T^* \cup \{t_x,t_y\})$ \;
		$k_1 = k_1 - 1$; $k_2 = k_2 - 1$\;
	}
	\While {$(|T^*| < k)$}{
		\If{$(k_1 > 0)$}{
			$k' = | T^* | + 1$ \;
			\For{$e=(t_x,t_y), (t_x\in T^{*^+} ,t_y \in T^+ \setminus T^*)$}{
		\textbf{if}  {\textsc{rel}($T^* \cup \{t_x,t_y\}$) $\geq \beta \cdot \textsc{rel}_{max}^{T,k'}$}{
						\textbf{then} \Compute{$\vartheta_{DC}(T^* \cup \{t_x,t_y\})$}\;}
					}
					$T^* = T^* \cup \argminl_{t_x\in T^{*^+},t_y \in T^+ \setminus T^*} \vartheta_{DC}(T^* \cup \{t_x,t_y\})$ \;
		
		}
		\If{$(k_2 > 0)$}{
			$k' = | T^* | + 1$\; 
			\For{$e=(t_x,t_y), (t_x\in T^{*^-} ,t_y \in T^- \setminus T^*)$}{
		\textbf{if}  {\textsc{rel}($T^* \cup \{t_x,t_y\}$) $\geq \beta \cdot \textsc{rel}_{max}^{T,k'}$}{
						\textbf{then} \Compute{$\vartheta_{DC}(T^* \cup \{t_x,t_y\})$}\;}
					}
					$T^* = T^* \cup \argminl_{t_x\in T^{*^-},t_y \in T^- \setminus T^*} \vartheta_{DC}(T^* \cup \{t_x,t_y\})$ \;
		
		}
	}
	\Return{$T^*$}
	\caption{DC-TA Algorithm (A-DC-TA)}
	\label{algo:Greedy3}
\end{algorithm}

\begin{figure*}[ht]
\begin{minipage}[t]{0.32\linewidth}
\centering
\includegraphics[height = 40mm,width = 45mm]{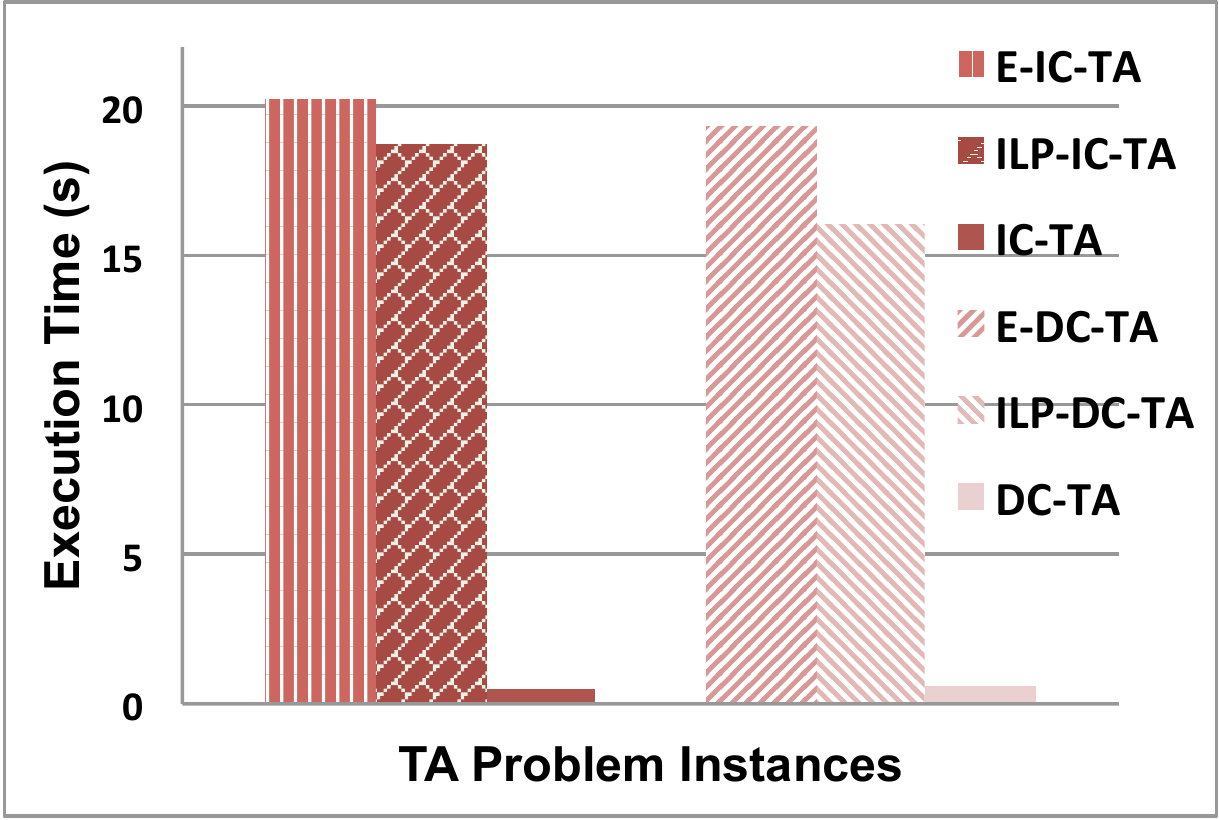}
\caption{{Execution time of TA algorithms with $k=10$, $\alpha=0.5$, $\beta=0.5$}}
\label{allExeTime}
\end{minipage}
\hspace{1mm}
\begin{minipage}[t]{0.32\linewidth}
\includegraphics[height = 40mm, width = 45mm]{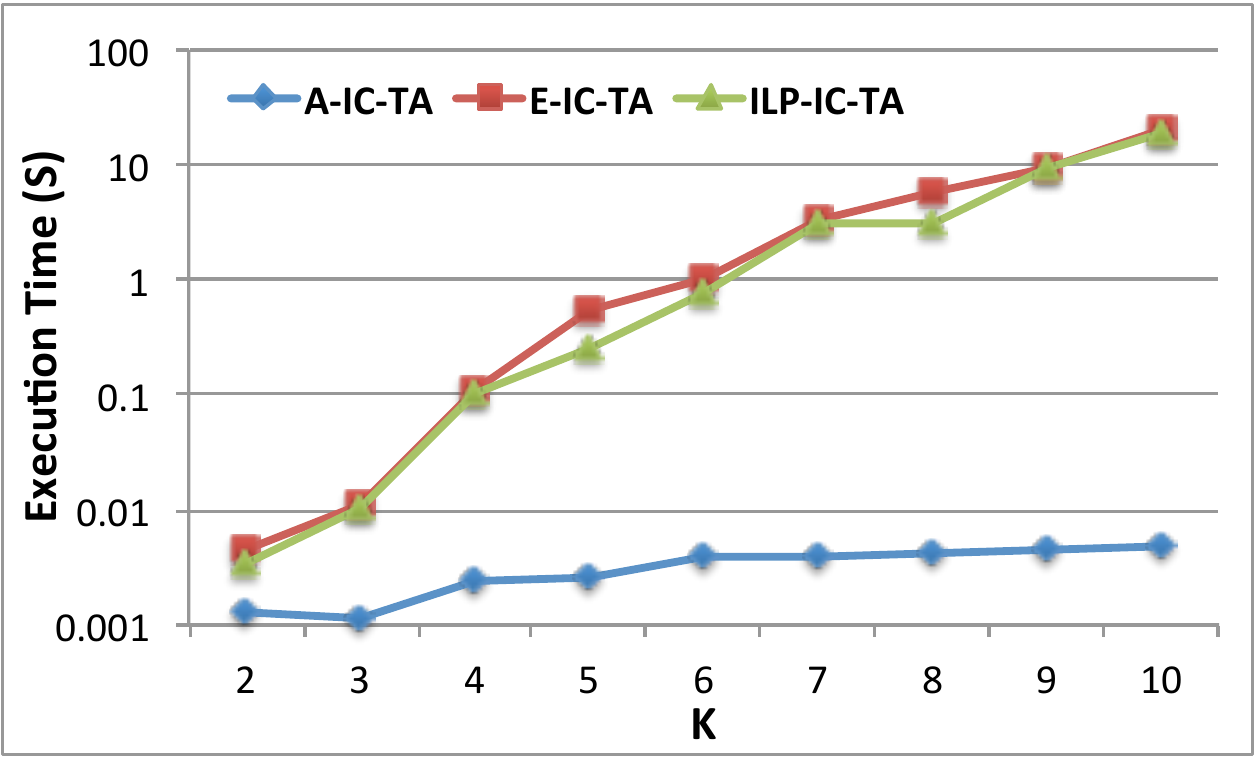}
\caption{{Execution time of IC-TA algorithms by varying $k$, $\alpha=0.5$, $\beta=0.5$}}
\label{fig:ICExeTime}
\end{minipage}
\hspace{1mm}
\begin{minipage}[t]{0.32\linewidth}
\includegraphics[height = 40mm,width = 45mm]{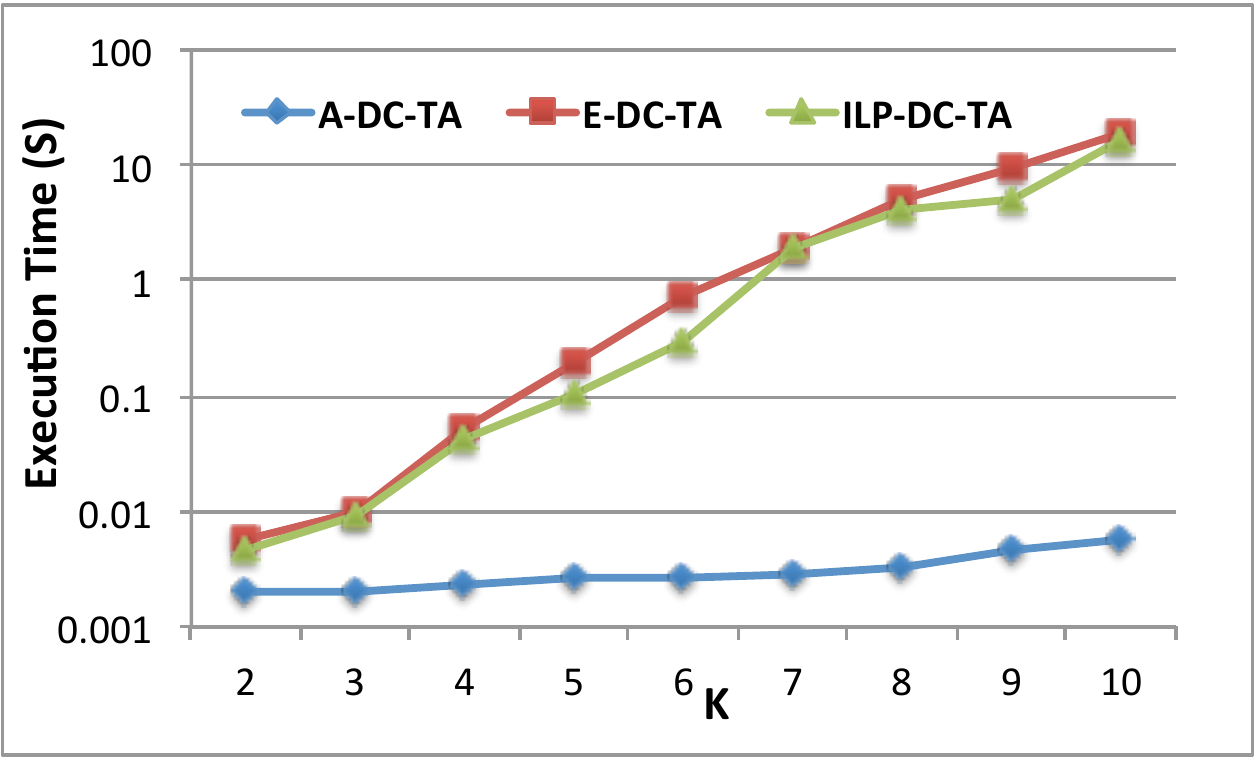}
\caption{{Execution time of DC-TA algorithms by varying $k$, $\alpha=0.5$, $\beta=0.5$}}
\label{fig:DCExeTime}
\end{minipage}
\end{figure*}

\section{Experiments}
\label{sec:expt}
\subsection{Experimental setup}
\vspace{0.01in}
\noindent
\textbf{System configuration}: Our prototype system is implemented in Java with JDK 5.0. All experiments were conducted on an Ubuntu machine with 2.0Ghz Intel processor and 8GB RAM. All numbers are obtained as the average over 10000 runs.

\vspace{0.01in}
\noindent {\bf Datasets:} We conduct a comprehensive set of experiments using both synthetic and real data crawled from the web to evaluate efficiency and quality of our proposed algorithms. For synthetic data, we generated a large boolean matrix of item attributes with positive and negative tags. For real data, we crawled Yahoo! Autos, Walmart and Google Product for building a car dataset and a camera dataset. We use the synthetic dataset for quantitative experiments, and the real dataset for qualitative study. The details of each dataset is described below:

\vspace{0.02in}
\noindent
{\em Synthetic Dataset}: 
We generate a large boolean matrix of dimension 1 million (items)$\times$ $200$ ($100$ attributes + $50$ positive tags + $50$ negative tags). We split the 100 independent and identically distributed attributes into four groups, where the value is set to 1 with probabilities of 0.75, 0.15, 0.10 and 0.05 respectively. For each of the 50 tags, we randomly picked a set of attributes that are correlated to it. A tag is set to 1 if majority of the attributes in its correlated set of attributes have boolean value 1. 

\vspace{0.02in}
\noindent
{\em Real Camera Dataset:} 
We crawl a real dataset of over hundred cameras listed at Walmart \footnote{www.walmart.com}. The Walmart camera data consists of 12,600 reviews from 11,500 users on 140 cameras. Since the camera information crawled from Walmart lacked well-defined item attribute values for all the cameras, we look up Google Products\footnote{www.google.com/about/products} and parse a total of 120 attributes such as {\sf \small self-timer}, {\sf \small red-eye fix}, {\sf \small auto focus}, {\sf \small built-in flash}, etc. We process the reviews to identify a set of positive and negative tags such as {\tt stunning photo quality}, {\tt great pocket camera}, {\tt short battery life}, {\tt expensive}, etc. using the keyword extraction toolkit AlchemyAPI\footnote{www.alchemyapi.com} which, in turn, uses natural language processing technology and machine learning algorithms to extract semantic meta-data from content. We employ RIPPER~\cite{Cohen:1995} to predict the set of rules that shows the dependency between item attributes and tags.

\vspace{0.02in}
\noindent
{\em Real Car Dataset:} 
We crawl a real dataset of 100 used cars listed at Yahoo! Autos\footnote{autos.yahoo.com} for the year 2010. The products contain technical specifications as well as ratings and reviews, which include pros and cons. We parse a total of 47 attributes: 15 numeric, and 32 boolean and categorical (the latter is generalized to boolean). The total number of reviews, i.e., pros and cons by users for the 100 cars is 2350. Since a feedback is labelled `pro' or `con', we do not need to employ any external text mining toolkit for getting the sentiments. The feedbacks are short phrases and keywords. These phrases are processed by domain experts to identify 20 representative positive and 20 representative negative tags that cover all the keywords crawled. For example, the `pro' keywords {\em driver seat comfort}, {\em cockpit comfort including ability to reach all controls easily}, {\em comfort is truly exceptional}, {\em super comfy and roomy for 4 people and dog} correspond to the representative positive tag {\tt comfortable}. 

\vspace{0.01in}
\noindent
\textbf{Performance Measures}: 
Our quantitative performance indicators are (i) {\em efficiency} of the algorithms, (ii) {\em approximation factor} of results produced by the approximation algorithms, and (iii) {\em quality} of the results produced. The efficiency of our algorithms is measured by the overall execution time, whereas approximation factor is determined by the ratio of the approximate result score to the actual optimal result score. The quality of result is measured by the ratio of features covered by our algorithms to the total number of features. We show that our algorithms are scalable and achieve much better response time than the exact algorithm while maintaining similar result quality.
In order to demonstrate that the top-$k$ tags returned by our approaches are useful to the end users, we conduct a user study through Amazon Mechanical Turk as well as write interesting case study.

\subsection{Experimental Results}
\begin{figure}
\centering
\includegraphics[height = 40mm,width = 50mm]{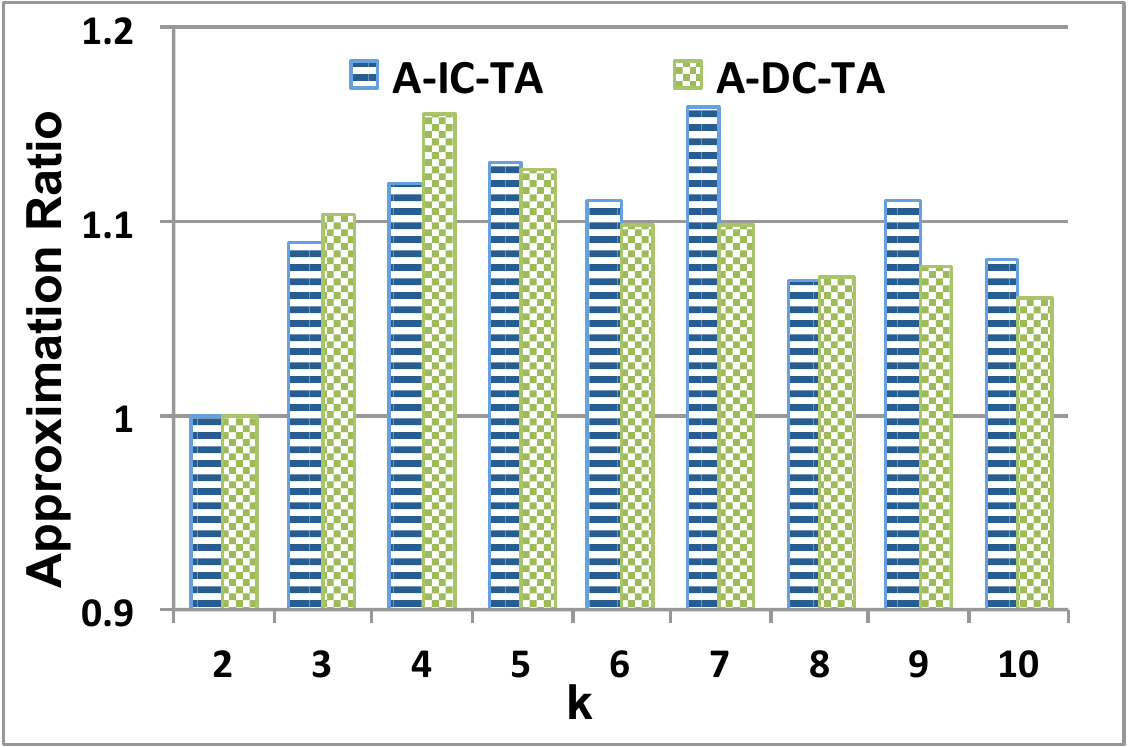}
\caption{{Approximation ratio of A-DC-TA and A-DC-TA $k$, $\alpha=0.5$, $\beta=0.5$}}
\label{fig:approxFactor}
\end{figure}

\subsubsection{Quantitative Evaluation}
We first compare the execution time of our approximation algorithms against the brute-force counterparts. We used an of-the-shelf ILP solver using GLPK in CVXOPT in Python\footnote{http://cvxopt.org/} to solve the ILP version of the problems.
Figure~\ref{allExeTime} shows that the execution time of the proposed algorithms A-IC-TA and A-DC-TA are several orders of magnitude faster than the corresponding exact algorithms E-IC-TA, ILP-IC-TA, E-DC-TA, and ILP-IC-TA for $k$=10, user factor $\alpha$=0.5, and relevance parameter $\beta$=0.5 on entire synthetic data.
Figures~\ref{fig:ICExeTime} and~\ref{fig:DCExeTime} compare execution time of A-IC-TA with E-IC-TA, ILP-IC-TA and that of A-DC-TA with E-DC-TA, ILP-DC-TA respectively by varying parameter $k$, with $\alpha$=$0.5$, and $\beta$=$0.5$. An interesting observation is that the cost of ILP does not always increase with $k$, possibly because ILP solver is based on branch and bound paradigm and the pruning of the search space is more efficient for some instances than for others. Moreover, we observe that by increasing $k$, execution time of the exact algorithms grow exponentially, while A-IC-TA and A-DC-TA scales well.

Next, we investigate the ratio of the approximate result score to the actual optimal result score. In A-IC-TA and A-DC-TA the approximation ratio is the value of the $\textsc {cov}_{IC}(T^*)$ and $\vartheta_{DC}(T^*)$ in Equation~\ref{eq:DC}  to the optimal solutions. We proved in theorems~\ref{th:IC-approxRatio} and~\ref{th:core_pol_and_approx}, A-IC-TA, and  A-DC-TA produce solutions with $2$-approximation of the optimal. Figure~\ref{fig:approxFactor} shows that by varying $k$ ,the approximation ratios are less than $2$. 

Finally, we evaluate the quality of results returned by our approximation algorithms by measuring the proportion of tags covered by the result set of $k$ tags in $T^*$. We compare the proposed algorithms A-IC-TA and A-DC-TA with the exact algorithms ILP-IC-TA and ILP-DC-TA by using the Independent-Coverage function, $\textsc {cov}_{IC}(T^*)$, in Equation~\ref{eq:UC} and Dependent-Coverage function, $\textsc {cov}_{DC}(T^*)$, in Equation~\ref{eq:DC} respectively. We conduct our experiments with different set of constraint conditions, i.e., user factor ($\alpha$), relevance parameter ($\beta$), and $k$. First, we set $\alpha=0.5$, $\beta=1.0$, and vary $k$ from $2$ to $10$ in Figures~\ref{fig:IC-k} and~\ref{fig:DC-k}. The results show that by increasing number of tags $k$, the proportion of covered item attribute values  are increased. Moreover, the quality of our A-IC-TA and A-DC-TA algorithms are almost same as exact algorithms ILP-IC-TA and ILP-DC-TA. Second, we set $k=10$, $\alpha=0.5$, and relevance parameter $\beta$ varies from $0.1$ to $0.9$ in step of $0.2$. Th results in Figure~\ref{fig:IC-rel} and~\ref{fig:DC-rel} show that although the relevance is increasing, proposed A-IC-TA and A-DC-TA algorithms are able to find $10$ tags with as high quality as the exact algorithms. Third, we set $k=10$, $\beta=0.5$, and user factor $\alpha$ varies from $0.1$ to $0.9$ in step of $0.2$. Results are shown in Figure~\ref{fig:IC-user} and~\ref{fig:DC-user}. As one can see from the figure, by increasing the user factor parameter the proportion of covered item attribute values is decreasing. In other words, there are some item attribute values that will be covered by negative tags and since the user factor is high, the lower negative tags are appeared which lead to lower quality. However, the results show that the quality of our algorithm is still as good as the exact algorithms. In summery, all the results from different set of constraint conditions confirm the fact that despite the significant reduction in execution time, our A-IC-TA and A-DC-TA algorithms do not compromise much in terms of analysis quality.

\begin{figure*}[ht]
\begin{minipage}[t]{0.32\linewidth}
\includegraphics[height = 40mm,width = 45mm]{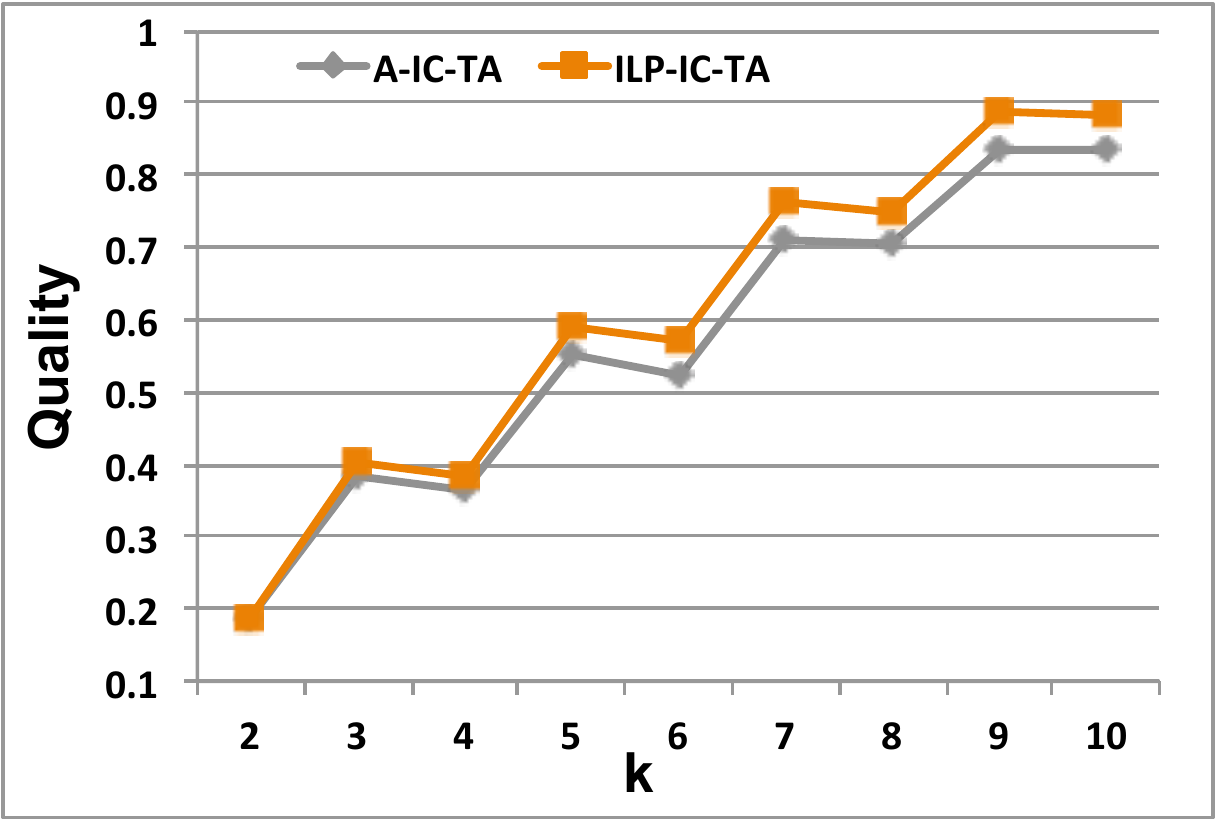}
\caption{{Quality of IC-TA algorithms by varying $k$, $\alpha=0.5$, $\beta=0.5$}}
\label{fig:IC-k}
\end{minipage}
\hspace{0.5mm}
\begin{minipage}[t]{0.32\linewidth}
\includegraphics[height = 40mm,width = 45mm]{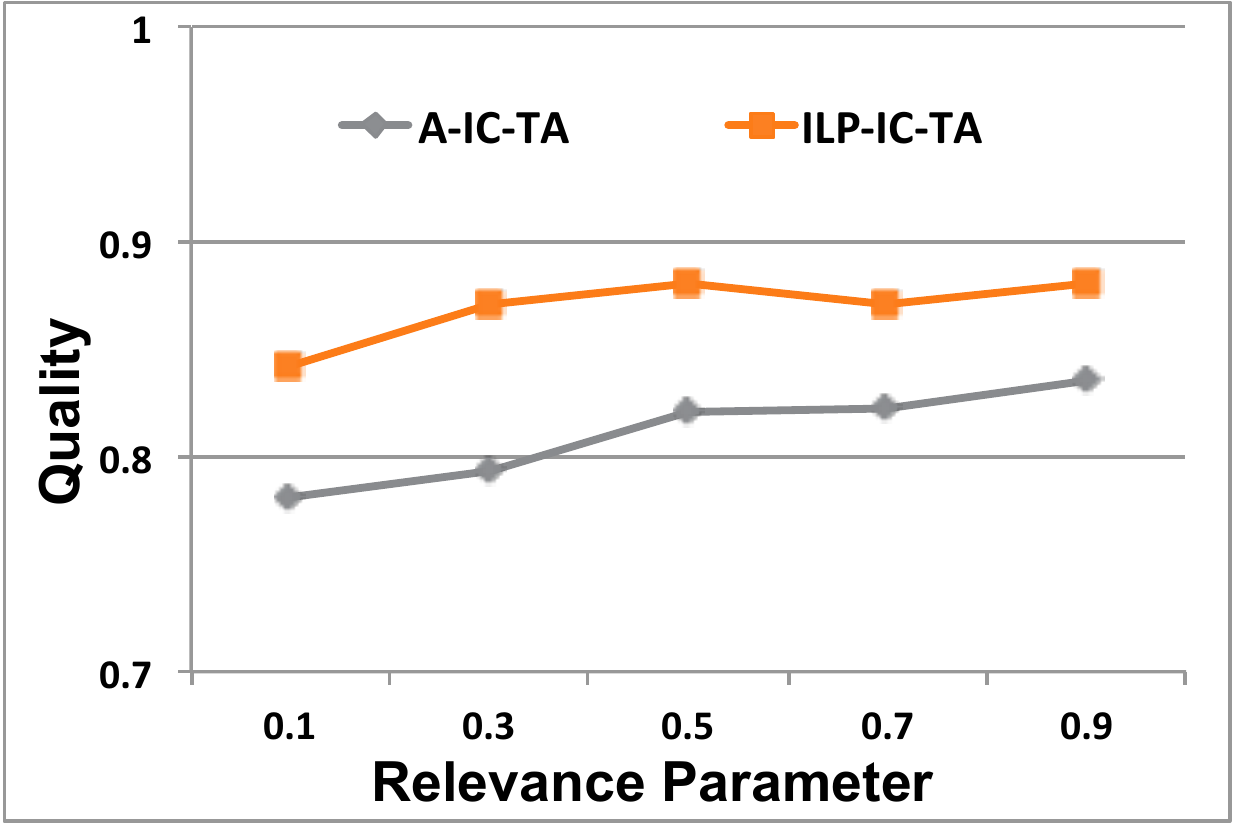}
\caption{{Quality of IC-TA algorithms by varying relevance parameter ($\beta$), $k=10$, $\alpha=0.5$}}
\label{fig:IC-rel}
\end{minipage}
\hspace{1mm}
\begin{minipage}[t]{0.32\linewidth}
\includegraphics[height = 40mm,width = 45mm]{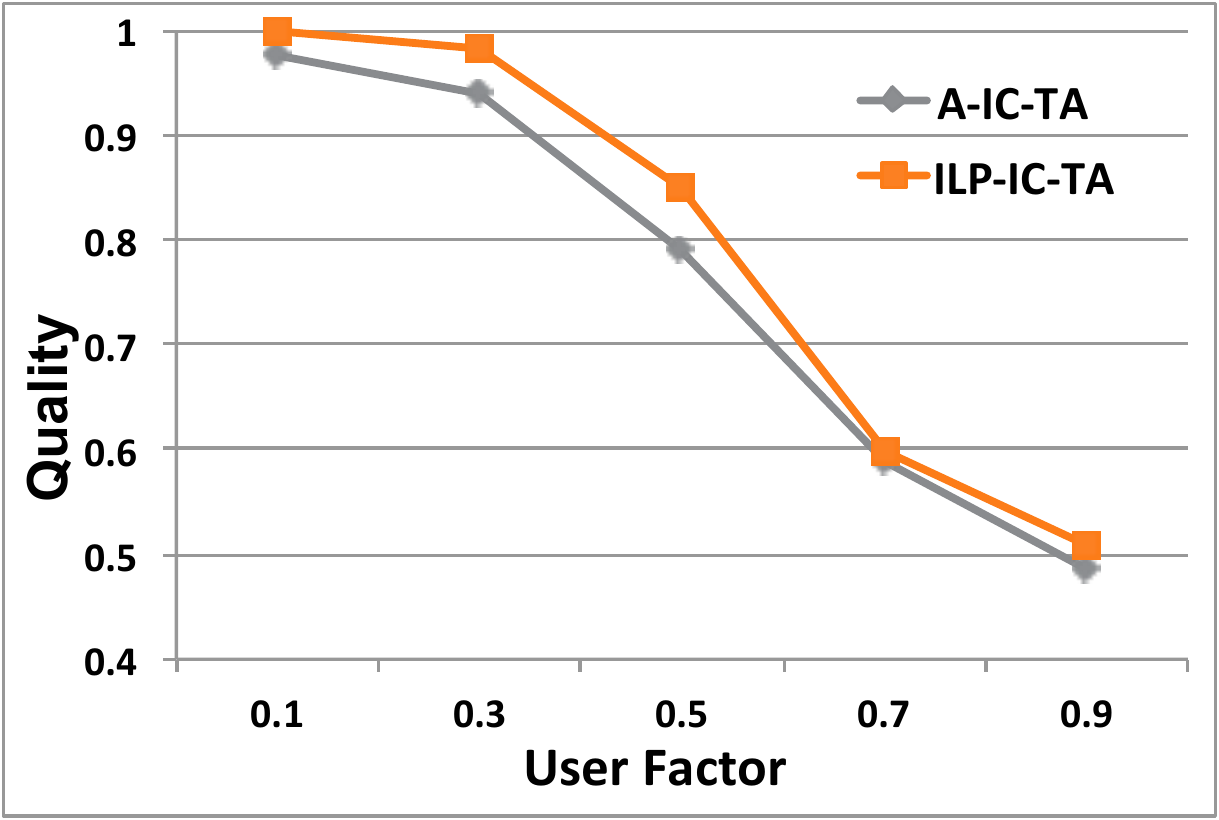}
\caption{\small{Quality of IC-TA algorithms by varying user factor ($\alpha$), $k=10$, $\beta=0.5$}}
\label{fig:IC-user}
\end{minipage}
\end{figure*}

\begin{figure*}[ht]
\begin{minipage}[t]{0.32\linewidth}
\includegraphics[height = 40mm,width = 45mm]{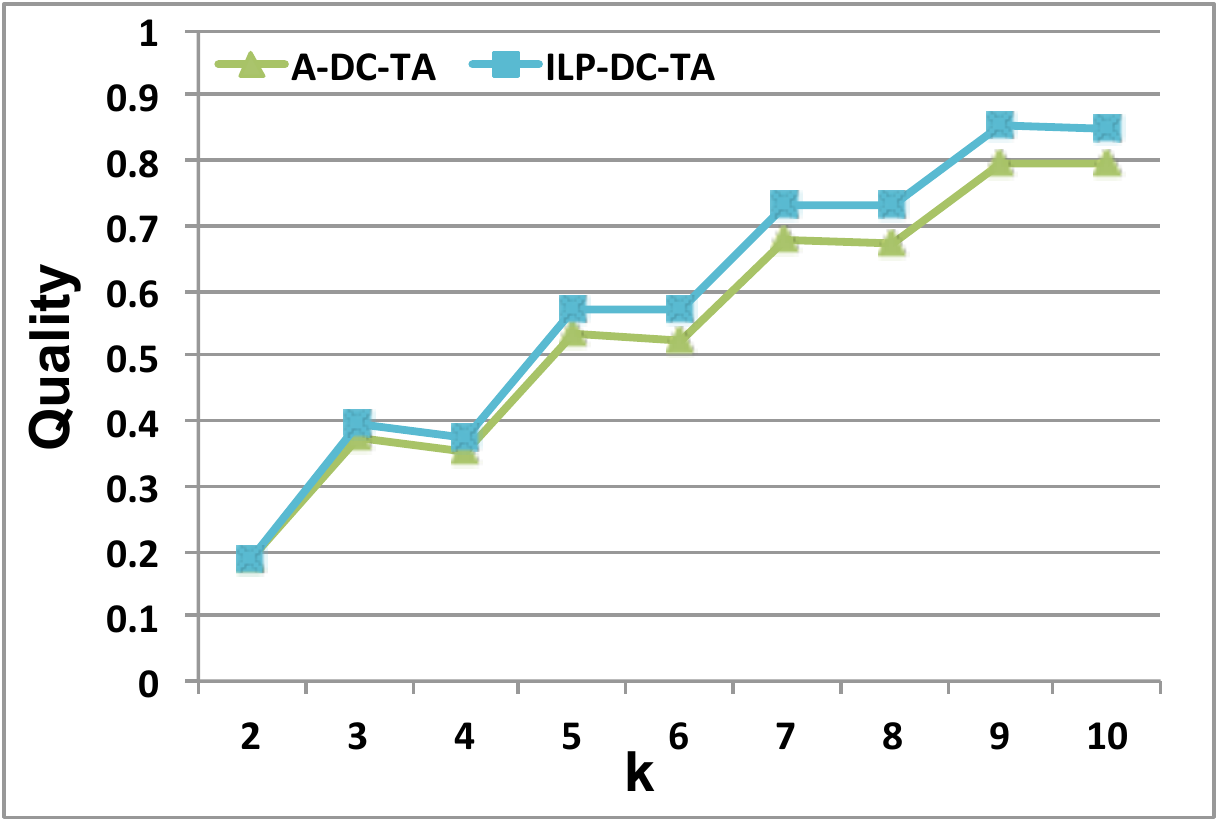}
\caption{{Quality of DC-TA algorithms by varying $k$, $\alpha=0.5$, $\beta=0.5$}}
\label{fig:DC-k}
\end{minipage}
\hspace{0.5mm}
\begin{minipage}[t]{0.32\linewidth}
\includegraphics[height = 40mm,width = 45mm]{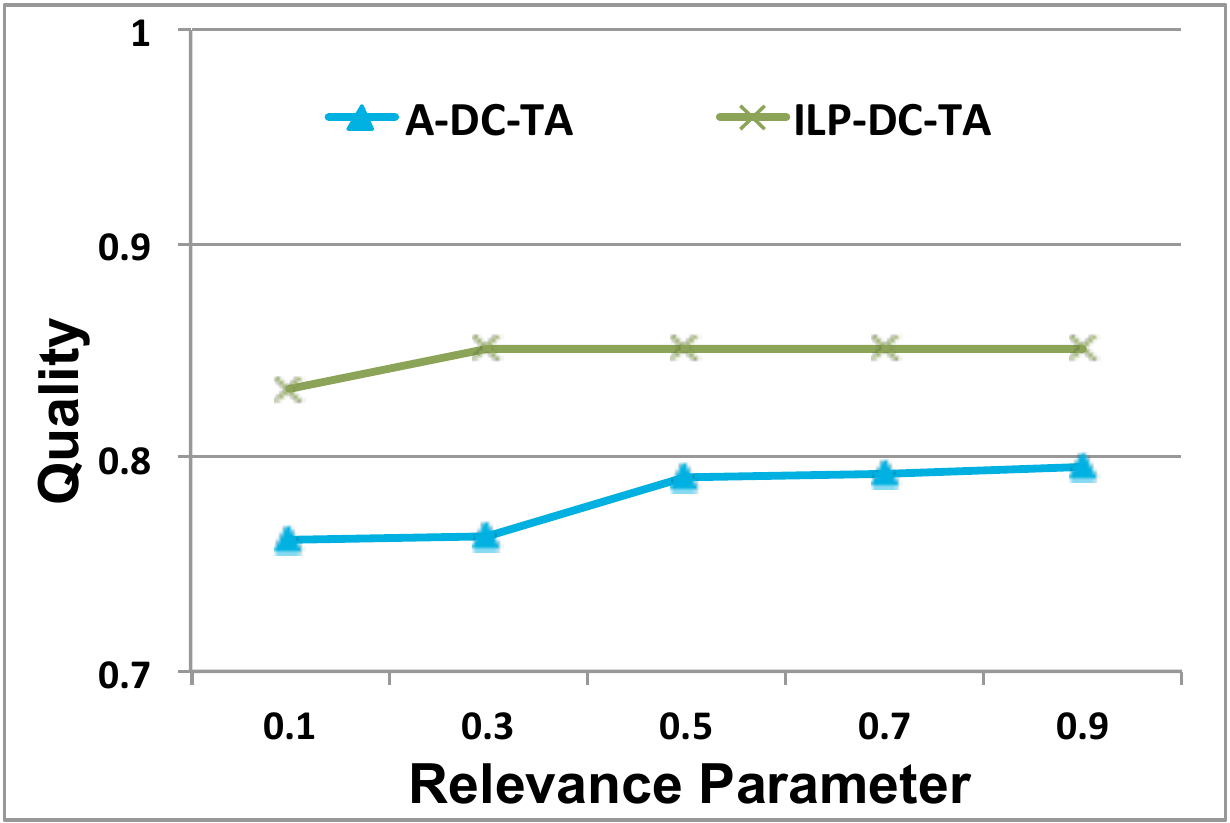}
\caption{{Quality of DC-TA algorithms by varying relevance parameter ($\beta$), $k=10$, $\alpha=0.5$}}
\label{fig:DC-rel}
\end{minipage}
\hspace{1mm}
\begin{minipage}[t]{0.32\linewidth}
\includegraphics[height = 40mm,width = 45mm]{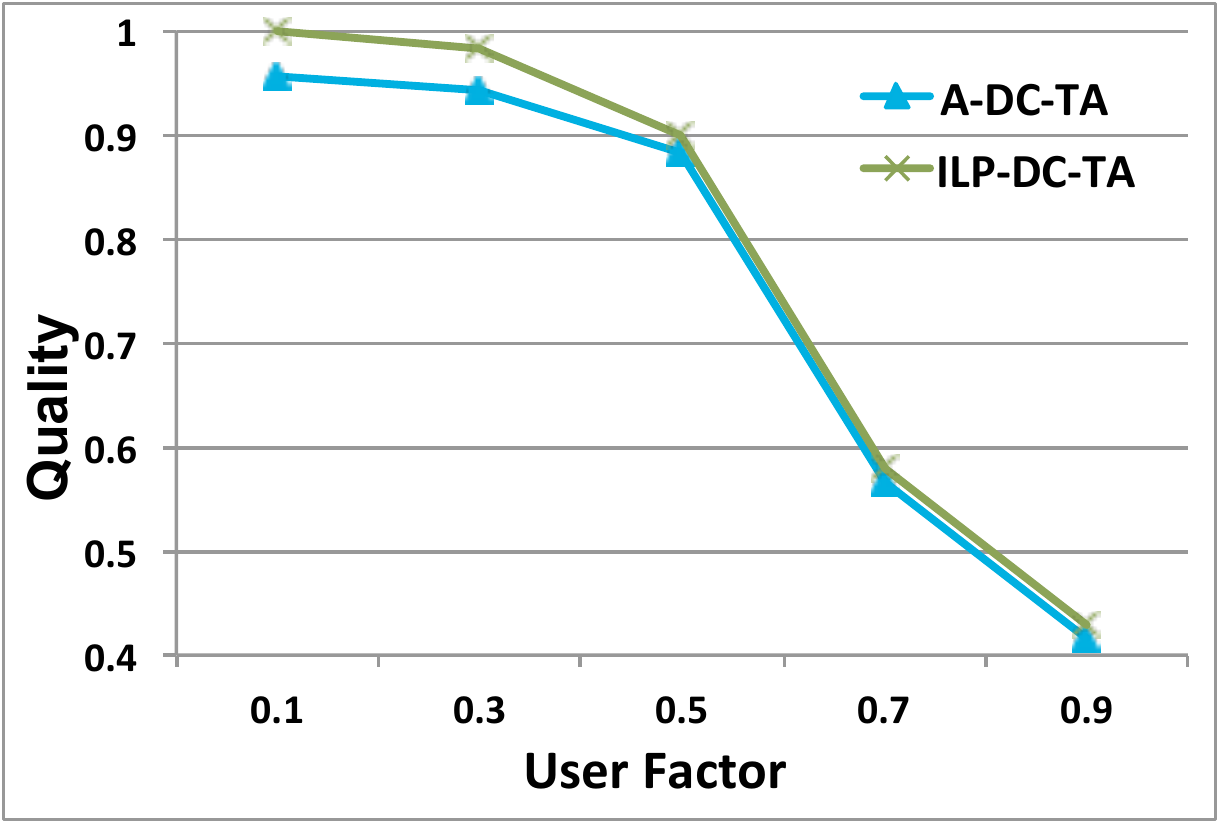}
\caption{\small{Quality of DC-TA algorithms by varying user factor ($\alpha$), $k=10$, $\beta=0.5$}}
\label{fig:DC-user}
\end{minipage}
\end{figure*}




\subsubsection{Qualitative Evaluation}
We now validate how users prefer tags returned by TagAdvisor over writing reviews from scratch in a user study conducted on Amazon Mechanical Turk\footnote{www.mturk.com} on the real camera dataset. We also present an interesting anecdotal result returned by our algorithm for an entry in the real car dataset.

\vspace{0.10in}
\noindent
{\bf User Study:}
We conduct a user study through Amazon Mechanical Turk (AMT to investigate if users prefer and benefit from our TagAdvisor system. We generate the top-$k$ tags for six cameras spanning different bands (Nikon, Canon, and Sony), and different types (digital SLR and compact point-and-shoot). The key objectives are: (i) to elicit the users' responses to the tags returned by our system \textemdash if they find the tags meaningful and adequate to review the product or if they prefer articulating their own review; (ii) to elicit the users' response to the products \textemdash if the feedback left by the users match the tags returned by our system.

We have 30 independent single-user tasks for each of the objectives. Each task is conducted in two phases: User Knowledge Phase and User Judgment Phase. During the first phase, we estimate the user's familiarity about camera and digital photography in general, and the six cameras that are being reviewed. During the second phase, we collect responses to our questions in the study from the users who are estimated to have a reasonable background in the first phase. For the study involving the second objective, we consult domain experts to validate if the tags submitted by the users for the cameras are similar to the tags returned by our system. 
Here are our observations.
\begin{itemize}
\item As many as 80\% users confirmed that they have ever reviewed a product (or service) online, which is a high but understandable percentage since they are AMT workers -- 75\% of these users admitted that they do not write online reviews frequently.
\item 67\% of the users voted that they are knowledgeable about the six cameras (or, other similar cameras) that they have been asked to review in this study.
\item An overwhelming 83\% of the users voted that they would submit online reviews more often if they are provided a set of meaningful keywords to choose from to express their feedback -- 80\% of these users clarified that their `Yes to TagAdvisor' response is also dependent on what tags are provided to them for this purpose.
\item 71\% of the users reviewed the six cameras choosing tags returned by TagAdvisor instead of writing the review from scratch.
\item Finally, 77\% of the users submitted feedback that matches tags returned by TagAdvisor -- 43\% of those users submitted tags that are similar to the ones returned by the Independent Coverage problem while the rest 57\% wrote tags that are similar to the ones returned by the Dependent Coverage problem, thereby endorsing that both Independent Coverage and Dependent Coverage problem are equally important.
\item An interesting observation is that over 81\%  of users,  who submitted their own tags wrote primarily about the more external aspects of the camera such as price, weight, physical look, lens, zoom, etc. instead of providing detailed comments about the quality of image, video capability, ease of use, etc. This is understandable since they are AMT workers and may not have used the exact same camera(s) in their recent past to provide in-depth feedback.
\end{itemize}

This validates the utility and usefulness of our system.

\vspace{0.10in}
\noindent
{\bf Case Study:}
We use the real car dataset to validate that our algorithms return meaningful tags - which meet user's criticalness in reviewing, have sentiment attached to them, and also cover different aspects of the item - as opposed to the tags returned by existing tag recommendation systems~\cite{Belem:2013,DBLP:conf/kdd/FengW12,DBLP:conf/sigir/SongZLZLLG08}. Since ~\cite{Belem:2013} is the only tag recommender engine that returns tags that are relevant and diverse, we compare our result against it.

Suppose a user wants to submit her feedback for a 2010 {\sf \small Audi Q5}\footnote{Note that, our results are not influenced by, or biased towards, any brand in particular.} by choosing from a set of tags advised to her. If $k = 6$, the tags suggested by the tag recommender in~\cite{Belem:2013} are:

\vspace{0.05in}
\noindent 
{\tt amazing power}, {\tt comfortable}, {\tt convertible top with sunroof}, {\tt nice style}, {\tt good gas mileage}, {\tt great auto transmission}
\vspace{0.05in}

Although this approach returns tags that cover diverse aspect of car, i.e, {\sf \small Standard Engine}, {\sf \small Seats}, {\sf \small Sunroof}, {\sf \small Fuel Capacity}, and {\sf \small Standard Transmission}, it does not consider sentiment. All the 6 tags are positive.

Considering user factor parameter $\alpha = 0.5$, relevance parameter $\beta = 0.5$, our IC-TA algorithm returns the tags:

\vspace{0.05in}
\noindent  
{\tt great auto transmission}, {\tt good gas mileage}, {\tt nice style}, {\tt odd engine sound}, {\tt wind noise at high sp\-eeds}, {\tt uncomfortable rear seat} 
\vspace{0.05in}	

These tags not only covers same aspects of the car as above, i.e, {\sf \small Standard Transmission}, {\sf \small Fuel Capacity}, {\sf \small Standard Engine}, {\sf \small Sunroof}, and {\sf \small Seats}, but it also satisfies the user's criticalness in reviewing ($\alpha = 0.5$), by returning three positive and three negative tags - the first three in the set above being positive and the last three being negative.

Under the same parameter specifications as above, our DC-TA algorithm returns the tags:

\vspace{0.05in}
\noindent 
{\tt amazing power}, {\tt convertible top with sunroof}, {\tt co\-mfortable}, {\tt odd engine sound}, {\tt wind noise at high speeds}, {\tt uncomfortable rear seat}
\vspace{0.05in}

These tags not only cover different aspects of the car such as {\sf \small Standard Engine}, {\sf \small Sunroof}, and {\sf \small Seats} but also allows the user to provide both positive and negative feedback for the same feature. Specifically, {\tt amazing power}, {\tt odd engine sound} are positive and negative tags respectively for the car feature {\sf \small Standard Engine}. Two t	ags {\tt convertible top with sunroof}, {\tt wind noise at high speeds} are positive and negative tags for the car feature {\sf \small Sunroof}. The last pair of tags {\tt comfortable}, and {\tt uncomfortable rear seat} are positive and negative tags for the car feature {\sf \small Seats}. Thus, the user has the option to select positive and/or negative feedback about this feature when she submits her feedback.

\vspace{-0.1in}
\section{Related Work}
\label{sec:relwork}
\vspace{0.05in}
\noindent\textbf{Tag Recommendation}: Tag recommendation has been extensively studied in literature~\cite{Belem:2013,DBLP:conf/kdd/FengW12,DBLP:conf/sigir/SongZLZLLG08,Hu:2010,wang:2013}.
The authors in~\cite{Hu:2010} focused on user perspective and they proposed a probabilistic framework for solving the personalized tag recommendation, but without considering diversity.
Result diversification has been studied in tag recommendation domain by~\cite{wang:2013,Belem:2013}; however, they take into account the possible topics and their goal is to provide high coverage and low redundancy with respect to those topics. The authors in~\cite{Belem:2013} used the general probabilistic framework in~\cite{Agrawal:2009} to address relevance and coverage. However, they assumes topics are independent, upon which a tag can not be dependent to the combination of the topics. The authors in~\cite{DBLP:conf/sigir/SongZLZLLG08} deals with the automated process to suggest useful and informative tags based on historical information. In our problem, the tags are more feedback than information  about the resource and hence calls for additional properties like coverage of all item attributes as well as sentiment polarity in opinion of the user for the different attributes of the item. A recent work~\cite{DBLP:conf/kdd/FengW12} proposes an optimization-based graph method for personalized tag recommendation. Though it considers both user features and item features for tag recommendation, the ranking-based solution recommends popular tags related to one or few specific aspects of the product and may evoke the rich-get-richer phenomenon, which in-turn is orthogonal to our objective of coverage. For example, if the popular tags for a point and shoot digital camera are {\tt lightweight}, {\tt thin}, and {\tt portable}, the method would return them as the top tags even though they are all related to the \textsf{weight} of the product. We intend to return tags covering the different aspects of the product such as \textsf{weight}, \textsf{price}, etc. as well as the different sentiments in opinion such as {\tt light} \textsf{weight}, {\tt heavy} \textsf{weight}, {\tt low} \textsf{price}, {\tt high} \textsf{price}, etc. so that the user can submit her review objectively. The authors in~\cite{Hu:2010} focused on user perspective and they proposed a probabilistic framework for solving the personalized tag recommendation, but without considering diversity.

\vspace{0.05in}
\noindent\textbf{Review Mining}: There has been a considerable amount of work in review summarization, ranking and selection~\cite{DBLP:conf/ACMicec/GhoseI07,DBLP:conf/kdd/HuL04,DBLP:conf/kdd/LappasCT12,DBLP:conf/kdd/TsaparasNT11}; yet, none of them can be readily extended to handle our problem. 
Review summarization creates statistical descriptions (i.e., a short snippet of text by extracting few existing sentences) of the review corpus in order to extract the proportion of positive and negative opinions about different aspects of a product. However, none of the current work directly caters to our objective of identifying {\em personalized} (i.e., user and item specific) tags. We leverage item descriptions, user demographics, as well as user sentiment. 
Review ranking aims to produce a score for each review and then display the top-$k$ highest-scoring reviews to the user~\cite{DBLP:conf/ACMicec/GhoseI07}. More specifically,~\cite{DBLP:conf/ACMicec/GhoseI07} proposed two ranking mechanisms for ranking product reviews: consumer-oriented ranking mechanism ranks the reviews according to their expected helpfulness, and a manufacturer oriented ranking mechanism ranks the reviews according to their expected effect on sales. However, they do not seek coverage over the range of features that are important to users and hence may return redundant information. For example, the top reviews for a point and shoot digital camera may just mention how {\tt ultrathin} and {\tt portable} it is, and not mention anything about how it has {\tt poor battery life}. 
Review summarization identifies a subset of helpful reviews that collectively provide both the negative and the positive aspects of each commented feature~\cite{DBLP:conf/kdd/TsaparasNT11}. While these methods do manage to expand the coverage of features and hence, diversify, they fail to capture the statistical properties of the actual review corpus. For example, if majority of the reviews for a SLR digital camera mention how {\tt excellent video quality} it produces, that should be given higher weight than returning one positive and one negative opinion about the camera feature {\sf video quality}. While~\cite{DBLP:conf/kdd/LappasCT12} returns a characteristic set of reviews that respects the proportion of opinions on each feature (both positive and
negative), as observed in the underlying corpus, neither does it leverage user preferences, nor does it leverage user feedback for other {\em similar items} - both of which are {\em necessary} considerations of the set of tags returned by our problem.   


\vspace{0.05in}
\noindent\textbf{Rule Learning}: In this paper, we used existing techniques to find the rules of the complex dependencies among item attributes and the tags. Rule learning has been extensively studied and there are different techniques such as: rule based classifiers techniques like RIPPER~\cite{Cohen:1995,Liu:1998,Quinlan:1993}, learning-based techniques like Re-RX~\cite{Diederich:20088} \cite{Setiono:2008}. In rule base classifiers, rules can be extracted directly from data~\cite{Liu:1998,Cohen:1995} or it can be extracted from other classification models~\cite{Quinlan:1993}. In~\cite{Liu:1998}, association rule mining is used to extract the rules while in~\cite{Cohen:1995} rules are extracted sequentially and for one class at a time. The authors in~\cite{Quinlan:1987} describe a technique for transforming decision trees to succinct collection of  if-then rules. Authors in~\cite{Chiang:2001} studied how to reduce the number of final rules in decision tree; ~\cite{Sirikulviriya:2011} proposed a new method that can integrate rules from multiple trees in a random forest to improve the comprehensiveness of the extracted rules. There has been many prior work on extracting classification rules from Support Vector Machines (SVM)~\cite{Nunez:2002},~\cite{Costa:2005},~\cite{Barakat:2010}, and~\cite{Diederich:20088}. In~\cite{Nunez:2002} rules are extracted from ellipsoids and hyper-rectangles formed using clustering algorithms. The fuzzy rule extraction method~\cite{Costa:2005} utilizes trained SVs to generate rule from each SV for each class.


\vspace{-0.1in}
\section{Conclusion}
\label{sec:conc}
In this paper, we introduce the novel TagAdvisor problem that leverages available user feedback for items in online review sites to simplify the review writing task. Our framework returns top-$k$ tags relevant to the product a user is reviewing, have sentiment attached to them, and cover the diverse attributes of the product. To the best of our knowledge, our framework is the first to consider all three measures simultaneously in the context of tag mining. Our work is also the first to address the popular problem in the web - how to motivate users to review a product online - in a principled way. We formulate the problem as a general-constrained optimization goal. By adopting different definitions of coverage, we identify two concrete problem instances that enable a wide range of real-world scenarios. We show that these problems are NP-hard and develop practical algorithms with theoretical bounds to solve them efficiently. Our experiments validate the utility of our problem and demonstrate that our proposed solutions generate equally good quality results as exact brute-force algorithms with much less execution time. 


\bibliographystyle{ACM-Reference-Format-Journals}
\bibliography{tagreco}

\end{document}